\documentclass[onecolumn]{IEEEtran}
\IEEEoverridecommandlockouts
\usepackage[mathscr]{eucal}
\usepackage[cmex10]{amsmath}
\usepackage{epsfig,epsf,psfrag}
\usepackage{amssymb,amsmath,amsthm,amsfonts,latexsym}
\usepackage{amsmath,graphicx,bm,xcolor,url,overpic}
\usepackage{fixltx2e}
\usepackage{array}
\usepackage{verbatim}
\usepackage{bm}
\usepackage{algorithmic}
\usepackage{algorithm}
\usepackage{verbatim}
\usepackage{textcomp}
\usepackage{epstopdf}
\usepackage{subfigure}
\usepackage{tikz}
\usepackage{caption}
\newcommand{\openone}{\leavevmode\hbox{\small1\normalsize\kern-.33em1}} 

\catcode`~=11 \def\UrlSpecials{\do\~{\kern -.15em\lower .7ex\hbox{~}\kern .04em}} \catcode`~=13 

\allowdisplaybreaks[4]

\newcommand{\calE}{\mathcal{E}}

\newcommand{\calO}{\mathcal{O}}

\newcommand{\calT}{\mathcal{T}}

\newcommand{\ba}{\mathbf{a}}
\newcommand{\bA}{\mathbf{A}}
\newcommand{\bb}{\mathbf{b}}

\newcommand{\bG}{\mathbf{G}}
\newcommand{\bh}{\mathbf{h}}
\newcommand{\bH}{\mathbf{H}}

\newcommand{\bI}{\mathbf{I}}

\newcommand{\bM}{\mathbf{M}}

\newcommand{\bs}{\mathbf{s}}
\newcommand{\bS}{\mathbf{S}}

\newcommand{\bu}{\mathbf{u}}

\newcommand{\bv}{\mathbf{v}}

\newcommand{\bw}{\mathbf{w}}

\newcommand{\bx}{\mathbf{x}}

\newcommand{\by}{\mathbf{y}}

\newcommand{\bz}{\mathbf{z}}

\newcommand{\rma}{\mathrm{a}}

\newcommand{\rmd}{\mathrm{d}}

\newcommand{\bbN}{\mathbb{N}}

\newcommand{\bbP}{\mathbb{P}}

\newcommand{\bbR}{\mathbb{R}}

\DeclareMathAlphabet{\mathbsf}{OT1}{cmss}{bx}{n}
\DeclareMathAlphabet{\mathssf}{OT1}{cmss}{m}{sl}

\DeclareSymbolFont{bsfletters}{OT1}{cmss}{bx}{n}  
\DeclareSymbolFont{ssfletters}{OT1}{cmss}{m}{n}
\DeclareMathSymbol{\bsfGamma}{0}{bsfletters}{'000}
\DeclareMathSymbol{\ssfGamma}{0}{ssfletters}{'000}
\DeclareMathSymbol{\bsfDelta}{0}{bsfletters}{'001}
\DeclareMathSymbol{\ssfDelta}{0}{ssfletters}{'001}
\DeclareMathSymbol{\bsfTheta}{0}{bsfletters}{'002}
\DeclareMathSymbol{\ssfTheta}{0}{ssfletters}{'002}
\DeclareMathSymbol{\bsfLambda}{0}{bsfletters}{'003}
\DeclareMathSymbol{\ssfLambda}{0}{ssfletters}{'003}
\DeclareMathSymbol{\bsfXi}{0}{bsfletters}{'004}
\DeclareMathSymbol{\ssfXi}{0}{ssfletters}{'004}
\DeclareMathSymbol{\bsfPi}{0}{bsfletters}{'005}
\DeclareMathSymbol{\ssfPi}{0}{ssfletters}{'005}
\DeclareMathSymbol{\bsfSigma}{0}{bsfletters}{'006}
\DeclareMathSymbol{\ssfSigma}{0}{ssfletters}{'006}
\DeclareMathSymbol{\bsfUpsilon}{0}{bsfletters}{'007}
\DeclareMathSymbol{\ssfUpsilon}{0}{ssfletters}{'007}
\DeclareMathSymbol{\bsfPhi}{0}{bsfletters}{'010}
\DeclareMathSymbol{\ssfPhi}{0}{ssfletters}{'010}
\DeclareMathSymbol{\bsfPsi}{0}{bsfletters}{'011}
\DeclareMathSymbol{\ssfPsi}{0}{ssfletters}{'011}
\DeclareMathSymbol{\bsfOmega}{0}{bsfletters}{'012}
\DeclareMathSymbol{\ssfOmega}{0}{ssfletters}{'012}

\newcommand{\bmeta}{\bm{\eta}}

\DeclareMathOperator*{\argmax}{arg\,max}

\usepackage{thmtools, thm-restate}
\newtheorem{theorem}{Theorem} 
\newtheorem{lemma}[theorem]{Lemma}

\newtheorem{corollary}[theorem]{Corollary}
\newtheorem{definition}{Definition}

\usepackage{cite}
\usepackage{amsmath,amssymb,amsfonts}
\usepackage{algorithmic}
\usepackage{graphicx}
\usepackage{textcomp}
\usepackage{xcolor,hyperref}
\def\BibTeX{{\rm B\kern-.05em{\sc i\kern-.025em b}\kern-.08em
    T\kern-.1667em\lower.7ex\hbox{E}\kern-.125emX}}

\begin{document}

\title{On the Error Exponent of Approximate Sufficient Statistics for $M$-ary Hypothesis Testing} 


\IEEEoverridecommandlockouts
\author{\IEEEauthorblockN{Jiachun Pan\IEEEauthorrefmark{1}, Yonglong Li\IEEEauthorrefmark{1}, Vincent Y.~F.~Tan\IEEEauthorrefmark{1}, and Yonina C. Eldar\IEEEauthorrefmark{2}}   \thanks{This work is partially funded by an NRF Fellowship (R-263-000-D02-281). The paper was presented in part at the 2020 International Symposium on Information Theory (ISIT).}   

\vspace{0.5em}

\IEEEauthorblockA{\IEEEauthorrefmark{1}Department of Electrical and Computer Engineering, National University of Singapore}\\
\IEEEauthorblockA{\IEEEauthorrefmark{2}Faculty of Mathematics and Computer Science, Weizmann Institute of Science, Israel}\\
\IEEEauthorblockA{Emails: \url{pan.jiachun@u.nus.edu}, \url{elelong@nus.edu.sg}, \url{vtan@nus.edu.sg}, \url{yonina.eldar@weizmann.ac.il}}
}

\maketitle

\begin{abstract}
Consider the problem of detecting one of $M$ i.i.d.\ Gaussian signals corrupted in white Gaussian noise. Conventionally, matched filters are used for detection. We first show that the outputs of the matched filter form a set of asymptotically optimal sufficient statistics in the sense of maximizing the error exponent of detecting the true signal. In practice, however, $M$ may be large which motivates the design and analysis of a reduced set of $N$ statistics which we term {\em approximate sufficient statistics}. Our construction of these statistics is based on a small set of filters that project the outputs of the matched filters onto a lower-dimensional vector using a sensing matrix. We consider a sequence of sensing matrices that has the desiderata of  row orthonormality  and low coherence. We analyze the performance of the resulting maximum likelihood (ML) detector, which leads to an achievable bound on the error exponent based on the approximate sufficient statistics; this bound recovers the original error exponent when $N=M$. We compare this to a bound that we obtain by analyzing a modified form of the Reduced Dimensionality Detector (RDD) proposed by Xie, Eldar, and Goldsmith [IEEE Trans.\ on Inform.\ Th., 59(6):3858-3874, 2013]. We show that by setting the sensing matrices to be column-normalized group Hadamard matrices, the exponents derived are ensemble-tight, i.e., our analysis is tight on the exponential scale given the sensing matrices and the decoding rule. Finally, we derive some properties of the exponents, showing, in particular, that they increase linearly in the compression ratio $N/M$.
\end{abstract}


\section{Introduction}
\label{subsec:motivation}
Consider the scenario in which we would like to detect an unknown signal $s_i(t)$ taking on one of $M$ distinct possibilities $\{s_i(t): i= 1,2,\ldots, M\}$ with equal probability. Measurement inaccuracies result in  noisy observations given by the received signal $y(t) = s_i(t) + z(t)$. To detect the signal, one can pose this problem as an $M$-ary hypothesis test. When the noise is white and Gaussian and the signals have equal energy, it is  well  known \cite{vanTrees} that the optimal detector is a {\em matched filter} in which one first takes the inner product of $y(t)$ with each of the hypothesized signals $s_i(t)$ and then chooses the index $i$ with largest inner product. 
The set of inner products, which constitutes a set of {\em sufficient statistics} in the case of equal-energy signals, together with the aforementioned decoding rule, minimizes the probability of detecting an incorrect hypothesis. 

In practice, $M$ may   be prohibitively large, which implies that many inner products need to be computed. For example, one can imagine that we would like to learn the particular category that an image belongs to. The true category is one out of a large number of categories $M$. We observe  a coarsely subsampled set of its pixels, the number of which is the length of the discrete-time signal $s_i$. This motivates a new class of detectors---one that computes a set of {\em approximate sufficient statistics} of {\em reduced dimensionality}. Instead of using the vector of $M$ inner products, we consider $N\le M$ judiciously computed inner products and base our decision solely on this smaller set of statistics~\cite{Xie2013,eldar}. A few natural questions beckon. First, how can we compress the vector of length $M$ into one of length $N$ such that sufficient information is preserved for reliable   detection? Second, what are the fundamental limits of doing so?  

\subsection{Related Works} \label{sec:related}
As shown in~\cite{Xie2013,eldar}, this problem shares close similarities to the vast literature on compressed sensing and sparse signal recovery~\cite{eldar2012compressed}. This is because, as we show in Section~\ref{sec:problem formulation},  one can formulate our problem in terms of the recovery of the non-zero location of a length-$M$ $1$-sparse vector, where the location of the non-zero element indicates which hypothesis is active. There are numerous works that study the information-theoretic limits of sparse signal recovery. For example, Wainwright~\cite{wainwright} and Wang, Wainwright, and Ramchandran \cite{wainwright2} derived sufficient and necessary conditions for exact support recovery. 
Reeves and Gastpar~\cite{Reeves} showed that recovery with an arbitrarily small but constant fraction of errors is possible and that in some cases computationally simple estimators are near-optimal. 
Tulino {\em et al.}~\cite{tulino} studied   sparse support recovery when the measurement matrix satisfies a certain ``freeness'' condition rather than having i.i.d.\ entries.
Scarlett and Cevher \cite{Scarlett} provided general achievability and converse bounds characterizing the trade-off between the error probability and number of measurements. 
However, most of the existing works assume that the measurements are corrupted by white Gaussian noise. The chief difference is that in our setting, due to the reduction in the dimensionality of the vector formed by the matched filter, the effective noise is no longer white Gaussian; this complicates the problem significantly and requires the adaptation of new analytical techniques so that the resultant error probability is readily analyzable. 

 Our problem is similar to that of Xie, Eldar, and Goldsmith~\cite{Xie2013}. In~\cite{Xie2013}, the authors were concerned with detecting the identities of active users within a general fading framework by using a so-called {\em Reduced Dimensionality Detector} (RDD), which incorporates a subspace projection with a thresholding strategy. The RDD allows for the detection of multiple ``active" signals corresponding to the detection of a $k$-sparse vector for $k\geq 1$. The scheme that the authors proposed results in the error probability decaying polynomially in  the total number of users $M$. In this paper, we are mainly concerned with establishing {\em fundamental limits} for a simpler setting---namely, the detection of a $1$-sparse signal (as only one hypothesis is true). Different from the setting in~\cite{Xie2013} in which the signals are deterministic, in this paper, motivated by optimal codes for transmission over Gaussian channels,  we consider {\em random} signals that are generated according to  product Gaussian distributions.
In addition, since we assume that the hypotheses are equiprobable, the maximum likelihood decoder is optimal and we   show that its error probability decays exponentially fast in $M$.
We also provide  analysis leading to a tighter bound of the error probability of a modified version of the RDD of~\cite{Xie2013} for the detection of a $1$-sparse signal within our specific framework (in which the underlying signals are Gaussian). This error probability also decays exponentially fast. 
	
Finally, we mention that Hayashi and Tan~\cite{hayashi} also considered the fundamental limits of approximate sufficient statistics from the viewpoint of reconstructing parametric families of distributions. However, rather than focusing on a hypothesis testing-centric problem,~\cite{hayashi} considered the problem of approximating  parametric distributions with small average Kullback-Leibler (KL) or total variation (TV) ``distance'' from given distributions based on a compressed version of the sufficient statistics.

\subsection{Summary of Main Contributions}
There are three main contributions in this paper.
\begin{itemize}
\item First, by leveraging the ideas of \cite{Xie2013,Eldar2009}, we construct a set of approximate sufficient statistics for our detection problem. To do so, we  judiciously select an appropriate family of sensing matrices, one that is efficient in compressing the underlying data and yet amenable to error exponent analysis. The resultant noise upon compression is no longer white and Gaussian. We choose the sensing matrices to be column-normalized group Hadamard matrices because of their row orthonormality and low coherence properties as shown by Thill and Hassibi in~\cite{thill}. 
\item Second, we analyze the performance of this approximate sufficient statistic under the maximum likelihood (ML) decoding strategy by leveraging ideas from Gallager's  error exponent (reliability function) analysis for channel coding \cite{gallager}. We obtain an achievable bound on the error probability and error exponent; this bound recovers the original error exponent (and hence is tight) when $N=M$, i.e., when there is no compression.  Moreover, we prove that the achievable  bound is ensemble-tight when the sensing matrix is a column-normalized group Hadamard matrix---this means that on the exponential scale, there is no looseness in our analyses. The error exponent increases linearly with compression rate $\alpha$ (i.e., the ratio $N/M$) when $\alpha$ is small.
\item Finally, we obtain an ensemble-tight achievable error exponent for a modified version of the  RDD~\cite{Xie2013} for the detection of $1$-sparse signals.  We show that the error exponent of the modified RDD is strictly smaller than that of the ML decoder, but the former is more computationally tractable.
\end{itemize}
\subsection{Paper Outline}
The rest of the paper is structured as follows. The problem is formulated in  Section~\ref{sec:problem formulation} in which we introduce  the notion of approximate sufficient statistics and various detectors. In Section~\ref{sec:main results}, we state our main results on achievable error exponents based on approximate sufficient statistics and discuss their implications. In Section~\ref{sec:proof}, we provide the proofs of our main results. We conclude and present some open problems in Section~\ref{sec:conclusion}. Proofs of more technical derivations are relegated to the appendices.

\section{Approximate Sufficient Statistics} 
\label{sec:problem formulation}
We assume that there are $M\ge 2$ hypotheses where $M$ is large. Under hypothesis $i \in \{1,\ldots, M\}$, the transmitted  real   signal is $s_i(t)$ where $t$ ranges over a set of discrete times $\calT$. Our observations under hypothesis $i$ are given by the signal
\begin{equation}
\label{eqn:receive}
y(t)=s_i(t)+z(t), \quad t\in \calT,
\end{equation}
where $\{z(t):t\in \calT\}$ is zero-mean white Gaussian noise with variance $\sigma^2>0$. The assumption that the noise is white and Gaussian is motivated from practical communication scenarios in which one assumes that the underlying signal is corrupted by an additive white Gaussian noise (AWGN) channel. For the sake of tractability, we assume that the set of times in which we observe $s_i(t)$---namely the set $\calT$---is the finite set $\{1,2,\ldots, T\}$ (and $T$ is allowed to grow with~$M$). Motivated by optimal coding schemes over an AWGN channel \cite[Chap.~7]{Cov06}, we assume each signal $\{s_i(t): t \in\calT\}$ is independently generated from {\color{black}the product Gaussian distribution $\prod_{t \in \calT} \mathcal{N}(s_i(t); 0, \calE^2)$.} Thus, $\calE^2/\sigma^2$ can be regarded as the {\em signal-to-noise ratio} and will henceforth be denoted by $\mathrm{SNR}=\calE^2/\sigma^2$. 
We can write~\eqref{eqn:receive} as
\begin{align}
\label{eqn:y}
y(t)=\sum_{i=1}^M b_i s_i(t) +z(t), \quad t\in \mathcal{T}
\end{align}
where the $\{0,1\}$-valued vector $\mathbf{b} = [b_1, b_2,\ldots, b_M]$ has $\ell_0$ norm given by $1$, i.e., it is $1$-sparse, as only one hypothesis is in effect.  Our goal is to detect the true transmitted signal given the observations $\{y(t):t\in\mathcal{T}\}$. 


\subsection{Matched Filter (MF) and Asymptotically Optimal Sufficient Statistics }

 Conventionally \cite{vanTrees}, one applies a matched filter before making a decision based on the outputs of the filter. The matched filter computes the inner product $\langle y(\cdot), s_i(\cdot)\rangle$ for each $i$ and chooses the index with the largest inner product as the decoding strategy. These inner products can be collated within the $M$-dimensional vector
\begin{align*}
	\bv:=[\langle y(\cdot),s_1(\cdot)\rangle, \langle y(\cdot),s_2(\cdot)\rangle,\cdots,\langle y(\cdot),s_M(\cdot)\rangle]^T \in \mathbb{R}^M.
\end{align*}
Using the received signal model (\ref{eqn:y}), $\bv$ can be written as
\begin{equation}
	\label{eqn:matched filter}
	\bv=\bG\bb+\bw,
\end{equation}
where $\mathbf{G} \in \bbR^{M\times M}$ is the  {\em Gram matrix} of the signals with its $(i,j)$-element defined as 
$[\bG]_{ij}=\langle s_i(\cdot),s_j(\cdot)\rangle$, $\bw=[w_1,w_2,\ldots,w_M]^{T}$, $\bv=[v_1,v_2,\ldots,v_M]^T , v_i=\langle y(\cdot),s_i(\cdot)\rangle$, and $w_i=\langle s_i(\cdot), z(\cdot)\rangle$.  
	
This filtering and subsequent decoding strategy is optimal in the sense of minimizing the probability of error~\cite{vanTrees} when the noise is additive white Gaussian and the {\em signals have equal energy}. In this case, the matched filter produces a vector of $M$ scalars $\{\langle y(\cdot),s_i(\cdot)\rangle\}_{i=1}^M$ that constitutes a set of sufficient statistics. However, when we assume the signals are independently generated from   product Gaussian distributions, the signals do not have equal energy almost surely, which means the vector produced by the matched filter is, in general, not a set of sufficient statistics.
In Theorem~\ref{prop:base}, however, we prove that $\bv$ constitutes a set of {\em asymptotically optimal sufficient statistics} in the sense that its error exponent (based on choosing its largest index $\argmax_j v_j$) is the same as that for the {\em optimal} decoding strategy, namely, the {\em maximum likelihood} decoder (i.e., the one that declares that the true hypothesis is $\argmax_j\bbP(\by\mid j)$).
 We define the probability of error when the matched filter is used as 
\begin{equation*}
P_{\mathrm{MF}}(\mathrm{err}):=\bbP\left(\argmax_{j\in\{1,\cdots,M\}} v_j \neq 1\right),
\end{equation*}
where the subscript ``MF'' stands for ``matched filter'', i.e., the error probability based on choosing the largest element of the vector $\bv$. 

\begin{figure}[t]
	\centering
	\begin{tikzpicture}[scale=0.7]
			\node (s) at (0,-2) {$s_{1}(t)$};
			\node (w) at (2, -3.5) {$z(t)$};
			\node (oplus) at (2, -2) [inner sep = 0pt] {$\displaystyle\bigoplus$};
			\draw [-latex] (s) -- (oplus);
			\draw [-latex] (w) -- (oplus);
			\node (y) at (4,-2)  {$y(t)$};
			\draw [-latex] (oplus) -- (y);
			
			\node (otimes1) at (6,0) [inner sep = 0pt] {$\displaystyle\bigotimes$};
			\node (otimes2) at (6,-2) [inner sep = 0pt] {$\displaystyle\bigotimes$};
			\node (otimesm) at (6,-4.5) [inner sep = 0pt] {$\displaystyle\bigotimes$};
			\draw [-latex]  (5,-2) -- (5,0) -- (otimes1);
			\draw [-latex] (y) -- (otimes2);
			\draw [-latex] (5,-2) -- (5,-4.5) -- (otimesm);
			\node (s1) at (6,-1.1) {$s_{1}(t)$};
			\node (s2) at (6,-3.1) {$s_{2}(t)$};
			\node (sm) at (6,-5.6) {$s_{M}(t)$};
			\draw [-latex] (s1) -- (otimes1);
			\draw [-latex] (s2) -- (otimes2);
			\draw [-latex] (sm) -- (otimesm);
			\path (s2) -- (otimesm)  node [midway, sloped] {$\dots$};
			
			\node[draw] (sum1) at (8,0) {$\sum_{t=1}^T$};
			\node[draw] (sum2) at (8,-2) {$\sum_{t=1}^T$};
			\node[draw] (summ) at (8,-4.5) {$\sum_{t=1}^T$};
			\draw [-latex] (otimes1) -- (sum1);
			\draw [-latex] (otimes2) -- (sum2);
			\draw [-latex] (otimesm) -- (summ);
			\node (v1) at (10,0) {$v_1$};
			\node (v2) at (10,-2) {$v_2$};
			\node (vm) at (10,-4.5) {$v_M$};
			\draw [-latex] (sum1) -- (v1);
			\draw [-latex] (sum2) -- (v2);
			\draw [-latex] (summ) -- (vm);
			\path (v2) -- (vm)  node [midway, sloped] {$\dots$};
%
%
%
%
%
%
%
%
	\end{tikzpicture}
	\caption{The structure of the matched filter.}
	\label{Fig.illustrate1}
\end{figure}
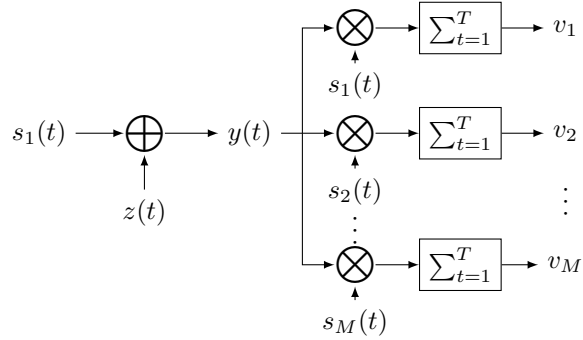

 \subsection{Approximate Sufficient Statistics}
 Motivated by the asymptotic optimality of the matched filter for the Gaussian signals $\{s_i(\cdot)\}_{i=1}^{M}$ (Theorem~\ref{prop:base}) and the analytical tractability of $\bv$ (since the signals $\{s_i(\cdot)\}_{i=1}^{M}$ are i.i.d.\ Gaussian), we now consider a low-dimensional projection of $\bv$ using ideas from compressive sensing and filter design. This lower-dimensional vector that is formed from correlating $y(\cdot)$ with $N\leq M$ filters $\{h_i(\cdot)\}_{i=1}^N$ leads to the notion of   {\em approximate sufficient statistics.}
The structure of the approximate sufficient statistics is illustrated in Figure~\ref{Fig.illustrate2}.
Following Eldar~\cite{Eldar2009}, the $N\le M$ correlating signals $\{h_i(t) : t\in\calT\}_{i=1}^N$ are chosen based on  the set of {\em biorthogonal signals} $\{\hat{s}_j(t) : t\in\calT\}_{j=1}^M$. Each biorthogonal signal $\hat{s}_i(t)$ is  defined as
\begin{equation*}
\hat{s}_j(t)=\sum_{l=1}^M[\bG^{-1}]_{jl} s_l(t), \quad 1\leq j \leq M.
\end{equation*}
The biorthogonality property implies that $\langle s_i(\cdot), \hat{s}_j(\cdot)\rangle=\delta_{ij}$ for all $i,j\in\{1,2,\ldots,M\}$. We remark that in discrete time, to construct the biorthogonal signals, we require $\mathbf{G}$ to be invertible \cite[Thm.~7.2.10]{Horn}; 
this necessitates that $T\ge M$.  We then choose the correlating signals to be
\begin{equation*}
\label{h}
h_i(t)=\sum_{j=1}^{M} a_{ij} \hat{s}_j(t), \quad 1\leq i \leq N
\end{equation*}
where $a_{ij}$ are some coefficients to be judiciously designed. Denote the $N\times M$ matrix $\bA$ with entries $[\bA]_{ij}=a_{ij}$ as the {\em sensing matrix}. The output of the $m$-th correlator is given by
\begin{align}
\label{eqn:um}
u_m&=\langle h_m(\cdot),y(\cdot)\rangle\nonumber\\*
&=\left\langle\sum_{j=1}^Ma_{mj}\hat{s}_j(\cdot), \sum_{i=1}^M b_i s_i(\cdot)\right\rangle+ \left\langle\sum_{j=1}^M a_{mj} \hat{s}_j(\cdot), z(\cdot)\right\rangle\nonumber\\*
&=\sum_{j=1}^M a_{mj}b_j +\eta_m,
\end{align}
where the effective output noise, which is no longer white, is given by 
\begin{equation}
\eta_m=\sum_{j=1}^M a_{mj}\langle\hat{s}_j(\cdot),z(\cdot)\rangle.\label{eqn:eta_m}
\end{equation}
 Defining the lower-dimensional vector $\bu=[u_1,u_2,\ldots,u_N]^T\in\mathbb{R}^N$ and the effective noise vector $\bmeta=[\eta_1,\eta_2,\ldots,\eta_N]^T\in\mathbb{R}^N$, we can express (\ref{eqn:um})  as
\begin{equation}
\label{eqn:ass}
\bu=\bA\bb+\bmeta.
\end{equation}
The new vector $\mathbf{u}$ of dimension $N\le M$ is an {\em approximate} sufficient statistic because $\bu$ can be viewed as a linear projection of the vector of asymptotically optimal sufficient statistic $\bv$ onto a low-dimensional subspace. Indeed, from~\eqref{eqn:ass}, the definition of $\eta_m$ in~\eqref{eqn:eta_m}, the definition of $\bv$ in~\eqref{eqn:matched filter},  and $w_i=\langle s_i(\cdot), z(\cdot)\rangle$,   we have
\begin{align*}
\bu & = \bA\bG^{-1}\bG \mathbf{b} + \bA\bG^{-1}\bw \\
&=\bA\bG^{-1} (\bG\mathbf{b}+ \bw)=\bA\bG^{-1}\bv.
\end{align*}
Thus, the above operations are in fact tantamount to compressing $\bv\in\bbR^M$ into $\bu\in\bbR^N$ and attempting to evaluate the tradeoff between the amount of compression $N/M$ versus the reliability as measured by the error exponent in decoding the true hypothesis given $\bv$ or $\bu$. 

We recall that we assume equal probability for each signal (or hypothesis). To detect the   signal based on the approximate sufficient statistic $\bu$, we will use the maximum likelihood detector, which minimizes the error probability given $\bu$. 
\subsubsection{Maximum Likelihood Detector}
\label{mld}
We define the maximum likelihood detector when given the approximate sufficient statistic $\mathbf{u}$ as
\begin{align}
\label{eqn:mld}
\phi_{\mathrm{ML}}(
\bu)=\argmax_{k\in\{1,2,\ldots,M \}}\bbP(\bu\mid k).
\end{align}
Here, with a slight abuse of notation, we denote the conditional density of observing $\bu$ in~\eqref{eqn:ass} given that hypothesis $k\in\{1,\ldots,M\}$ is in effect as $\mathbb{P}(\mathbf{u}\mid k)$. {\color{black} Since we assume that the signals $\{s_i(\cdot)\}_{i=1}^M$ and thus the Gram matrix $\bG$ are random, given a fixed hypothesis $k$, the distribution of $\bu$ cannot be expressed in closed-form and is computationally expensive to calculate. Indeed, by the law of total probability, it can be written as
\begin{align}
\label{eqn:totalpro}
\mathbb{P}(\mathbf{u}\mid k)=\int \mathbb{P}(\mathbf{u}\mid \bG,k)\mathbb{P}(\bG)\,\mathrm{d}\bG,
\end{align}
where
\begin{align}
\label{eqn:gaussian}
\mathbb{P}(\mathbf{u}\mid \bG,k)=\frac{\exp\{- \frac{1}{2}(\bu-\ba_k)^T(\sigma^2\bA\bG^{-1}\bA^T)^{-1}(\bu-\ba_k)\}}{\sqrt{(2\pi)^N\sigma^2|\bA\bG^{-1}\bA^T|}},
\end{align}
and 
the probability distribution function of the Gram matrix $\bG$ (i.e., $\bbP(\bG)$) is a Wishart distribution which we further elaborate on  in Section~\ref{sec:preli}. The distribution of $\bu$ given $\bG$ and the hypothesis $k$ is given as in~\eqref{eqn:gaussian} because it is Gaussian with mean $\ba_k$ and covariance matrix $\sigma^2\bA\bG^{-1}\bA^T$.}
In this case, the probability of detecting the underlying hypothesis (which recall is hypothesis $1$) incorrectly is 
\begin{equation*}
P_{\mathrm{a}}(\text{err}):=\bbP\left(\phi_{\mathrm{ML}}(\bu) \neq 1\right),
\end{equation*}
where the subscript ``a'' stands for ``approximate sufficient statistic''.

\subsubsection{Modified Reduced Dimensionality Detector (RDD) of Xie, Eldar, and Goldsmith \cite{Xie2013}}
We now introduce an alternative decoding strategy which is motivated by the so-called {\em Reduced Dimensionality Detector} of Xie, Eldar, and Goldsmith~\cite{Xie2013} and is given by the recipe
\begin{align}
\label{rdd}
\phi_{\mathrm{RDD}}(\bu)=\argmax_{k\in\{1,2,\ldots,M\}} |\mathbf{a}_k^T \mathbf{u}|,
\end{align}
where $\mathbf{a}_k$ is the $k$-th column of sensing matrix $\bA$. Note that in the problem setting of~\cite{Xie2013}, $b_i\in\{-1,0,1\}$ and to detect the active users (corresponding to the indices $i$ such that $b_i\in\{-1,1\}$), the absolute value in~\eqref{rdd} is necessary. In our setting, however, $b_i\in\{0,1\}$ and since we are trying to detect the hypothesis whose $b_i$ is equal to $1$, we design a   {\em modified} version of the RDD as follows:
\begin{align}
\label{rdd-like}
\phi_{\mathrm{mRDD}}(\bu)=\argmax_{k\in\{1,2,\ldots,M\}} \mathbf{a}_k^T \mathbf{u}.
\end{align}
Note that the absolute value operator is removed. 
We term this detector as a {\em Modified RDD} and the subscript ``mRDD" in~\eqref{rdd-like} reflects this terminology. The geometric interpretation of Modified RDD is that it is a linear detector that projects $\bu$ onto the subspace formed by the columns of $\bA$. The average error probability of Modified RDD is
\begin{align*}
P_{\mathrm{mRDD}}(\mathrm{err}):=\bbP\left(\phi_{\mathrm{mRDD}}(\bu)\neq 1\right).
\end{align*}

\subsection{Problem Statement}
In the following, we compare $P_{\mathrm{MF}}(\text{err})$ and the error probabilities of the two other proposed detectors $P_{\mathrm{a}}(\text{err})$ and $P_{\mathrm{mRDD}}(\text{err})$. For the sake of analytical tractability, and as conventionally done in information theory~\cite{gallager}, we compare these probabilities on the exponential scale as $M$ and $N$, which are assumed to scale linearly with each other, tend to infinity.  Thus, we compare the exponential decay rates (also called error exponents) corresponding to the various probabilities. 

{\color{black}When the prior probabilities for each hypothesis are equal and the maximum likelihood detector is used, we obtain the smallest probability of error based on the approximate sufficient statistics $\bu$.} Hence the comparison between $P_{\mathrm{MF}}(\text{err})$ and $P_{\mathrm{a}}(\text{err})$ demonstrates the degradation of the error probability (or error exponent) as a function of the amount of compression, quantified by $\alpha$. {\color{black}However, as explained in Section~\ref{mld}, computing $\bbP(\bu\mid k)$ is computationally expensive due to the randomness of the matrix $\bG$.} In contrast, computing $\phi_{\mathrm{mRDD}}$ in~\eqref{rdd-like} is computationally tractable as all that is needed is to take the inner product of $\bu$ with each of the $\ba_k$'s and then choose the maximum. Thus, by comparing $P_{\mathrm{a}}(\text{err})$ to $P_{\mathrm{mRDD}}(\text{err})$ again on the exponential scale, we are able to ascertain the degradation in the error probability (or the error exponent) when the more computationally efficient decoder---namely, the modified RDD---is employed.

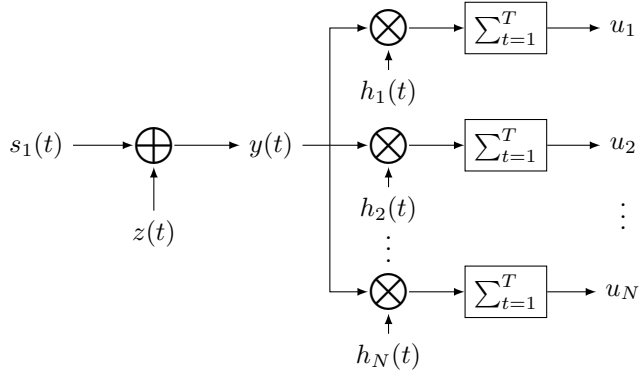
\begin{figure}[t]
	\centering
		\begin{tikzpicture}[scale=0.78]
		\node (s) at (0,-2) {$s_{1}(t)$};
		\node (w) at (2, -3.5) {$z(t)$};
		\node (oplus) at (2, -2) [inner sep = 0pt] {$\displaystyle\bigoplus$};
		\draw [-latex] (s) -- (oplus);
		\draw [-latex] (w) -- (oplus);
		\node (y) at (4,-2)  {$y(t)$};
		\draw [-latex] (oplus) -- (y);
		
		\node (otimes1) at (6,0) [inner sep = 0pt] {$\displaystyle\bigotimes$};
		\node (otimes2) at (6,-2) [inner sep = 0pt] {$\displaystyle\bigotimes$};
		\node (otimesm) at (6,-4.5) [inner sep = 0pt] {$\displaystyle\bigotimes$};
		\draw [-latex]  (5,-2) -- (5,0) -- (otimes1);
		\draw [-latex] (y) -- (otimes2);
		\draw [-latex] (5,-2) -- (5,-4.5) -- (otimesm);
		\node (s1) at (6,-1.1) {$h_{1}(t)$};
		\node (s2) at (6,-3.1) {$h_{2}(t)$};
		\node (sm) at (6,-5.6) {$h_{N}(t)$};
		\draw [-latex] (s1) -- (otimes1);
		\draw [-latex] (s2) -- (otimes2);
		\draw [-latex] (sm) -- (otimesm);
		\path (s2) -- (otimesm)  node [midway, sloped] {$\dots$};
		
		\node[draw] (sum1) at (8,0) {$\sum_{t=1}^T$};
		\node[draw] (sum2) at (8,-2) {$\sum_{t=1}^T$};
		\node[draw] (summ) at (8,-4.5) {$\sum_{t=1}^T$};
		\draw [-latex] (otimes1) -- (sum1);
		\draw [-latex] (otimes2) -- (sum2);
		\draw [-latex] (otimesm) -- (summ);
		\node (v1) at (10,0) {$u_1$};
		\node (v2) at (10,-2) {$u_2$};
		\node (vm) at (10,-4.5) {$u_N$};
		\draw [-latex] (sum1) -- (v1);
		\draw [-latex] (sum2) -- (v2);
		\draw [-latex] (summ) -- (vm);
		\path (v2) -- (vm)  node [midway, sloped] {$\dots$};
%
%
%
%
%
%
%
%
		\end{tikzpicture}
	\caption{The structure of the approximate sufficient statistics.}
	\label{Fig.illustrate2}
\end{figure}

\section{Main Results}
\label{sec:main results}
In this section, we present definitions and results for the {\em matched filter error exponent} and the {\em approximate sufficient statistic error exponent}. Let $\beta=T/M$ be the {\em transmission rate} (ratio of number of time samples to number of signals). Throughout the paper, we assume $1\le \beta<\infty$. 
The requirement that $\beta\geq1$ is mandated by a technical reason, namely, the fact that we require $\mathbf{G}$ to be invertible. 

\subsection{Matched Filter}
 
In this subsection, we define and present our results for the matched filter error exponent.
\begin{definition} [Matched filter error exponent]
	The matched filter error exponent with transmission rate $\beta$ is defined as 
	\begin{equation*}
	E_{\mathrm{MF}}(\beta,\mathrm{SNR})=\liminf_{M\rightarrow \infty}-\frac{1}{M}\log P_{\mathrm{MF}}(\mathrm{err}).
	\end{equation*}
\end{definition}

The following theorem whose proof leverages Gallager's bounding techniques~\cite[Sec.~7.4]{gallager} for proving bounds on the reliability function of a channel establishes the optimal error exponent as a function of the transmission rate $\beta$ and $\mathrm{SNR}$. The proof of Theorem~\ref{prop:base} is in Appendix~\ref{app:thm1}. 
\begin{theorem} \label{prop:base}
	Let $E^*(\beta, \mathrm{SNR})$ be the optimal error exponent, i.e., the error exponent when the optimal decoder (the maximum likelihood detector) is employed. 
	For $\beta\ge 1$, 
	\begin{align}
	E^*(\beta,\mathrm{SNR})=E_{\mathrm{MF}}(\beta,\mathrm{SNR})=\frac{\beta}{2}\log\left(1+\frac{\mathrm{SNR}}{2}\right).	\label{ss}
	\end{align}
\end{theorem}
In addition to providing an expression for $E_{\mathrm{MF}}(\beta, \mathrm{SNR})$, this theorem says that the matched filter is asymptotically optimal, i.e., $E_{\mathrm{MF}}(\beta,\mathrm{SNR})$ is equal to the error exponent of the optimal decoding strategy $E^*(\beta,\mathrm{SNR})$. This is intuitively plausible because even though the signals $\{s_i\}_{i=1}^M$ do not have the same energy (almost surely), their energies $\|s_i\|_2^2/T$ are arbitrarily close to one another with high probability when $T\to\infty$, so the matched filter is asymptotically optimal.

\subsection{Approximate Sufficient Statistic}
 In this subsection, we define and present our results for the error exponent associated with the approximate sufficient statistics that we construct in the following.
\begin{definition}
	Let $\{N_M\}_{M=1}^\infty$ be an increasing sequence of integers such that the limit $\lim_{M\to\infty} N_M/M$ exists.  The sequence of sensing matrices $\{\mathbf{A}_M \in \mathbb{R}^{N_M\times M}\}_{M=1}^{\infty}$ has {\em compression rate} $\alpha$ if $\lim_{M\to\infty}{N_M}/{M}=\alpha$.
\end{definition}
{\color{black}We also assume throughout the paper that $0< \alpha\le 1$.  This means that asymptotically,  the approximate sufficient statistics $\mathbf{u}$ has dimension $N$ no larger than $M$,  the number of hypotheses we started out with.  }
\begin{definition} [Approximate sufficient statistic error exponent]
	We say that $E_\rma$ is an {\em achievable error exponent} with transmission rate $\beta$ and compression rate $\alpha$ if there exists a sequence of sensing matrices $\{\bA_M\in \mathbb{R}^{N_M\times M}\}_{M=1}^{\infty}$ with compression rate $\alpha$  and a corresponding sequence of decoders (maps from the approximate sufficient statistic to a hypothesis) such that the detection error probability $P_{\rma}(\mathrm{err})$ satisfies 
	\begin{equation*}
	\liminf_{M\rightarrow \infty} -\frac{1}{M} \log P_{\rma}(\mathrm{err})\ge E_\rma.
	\end{equation*}
	We define the {\em approximate sufficient statistic error exponent} $E_{\mathrm{a}}^*(\alpha,\beta,\mathrm{SNR})$ as
	\begin{equation*}
	E_\rma^*(\alpha,\beta,\mathrm{SNR})=\sup E_\rma,
	\end{equation*}
	where the supremum is taken over all achievable error exponents with transmission rate $\beta$ and compression rate $\alpha$ when the signal-to-noise ratio of the system is $\mathrm{SNR}$.
\end{definition}


Theorems \ref{compressed_exp} and \ref{compressed_rdd} present lower bounds on the approximate sufficient statistic error exponent given a certain sequence of sensing matrices and two different decoding strategies.  In the compressive sensing literature, the sensing matrix is either deterministic or random~\cite{Foucart}. In our work, we find it convenient in our analysis to choose deterministic sensing matrices that have orthonormal rows and low coherence to obtain achievable lower bounds on the error exponent. We find that a column-normalized group Hadamard matrix~\cite{thill}, described in greater detail in Section~\ref{sec:preli}, satisfies these desirable properties. 
In particular, it has the favorable property of having a low coherence when $\frac{M-1}{N}$ is a natural number. We recall that the coherence of a matrix~\cite{coherence}, formally defined in~\eqref{eqn:coherence1}, is the largest (in magnitude) inner product between any two columns of the matrix. This low coherence property is particularly desirable in ensuring the analysis is amenable to obtaining  bounds on the error exponent under both maximum likelihood  decoding and the modified RDD.   
Before introducing our results, we first define error exponents pertaining to this choice of sensing matrix $\bA$. 

\begin{definition}[Maximum likelihood detector error exponent]
	The error exponent with transmission rate $\beta$ and compression rate $\alpha$ when we use the maximum likelihood detector and set $\bA$ to be  a column-normalized group Hadamard matrix with compression rate $\alpha$ is defined as 
	\begin{equation*}
	E_{\mathrm{a},\mathrm{ML}}(\beta,\alpha,\mathrm{SNR})=\liminf_{M\rightarrow \infty}-\frac{1}{M}\log P_{\mathrm{a},\mathrm{ML}}(\mathrm{err}),
	\end{equation*}
	where $P_{\mathrm{a},\mathrm{ML}}(\mathrm{err})$ is the detection error probability when we set $\bA$ to be a column-normalized group Hadamard matrix and the maximum likelihood detector in~\eqref{eqn:mld} is used.
\end{definition}

\begin{theorem}
	\label{compressed_exp}
	For $\beta\geq 1$ and $0<\alpha\le 1$ such that  $1/\alpha$ is a natural number, 
	\begin{align}\label{exponent_mld}
	E_{\mathrm{a}}^*(\beta,\alpha,\mathrm{SNR})&\geq E_{\mathrm{a},\mathrm{ML}}(\beta,\alpha,\mathrm{SNR})\\*
	&=\frac{\beta-1+\alpha}{2}\log\left(1+\frac{\alpha\mathrm{SNR}}{2}\right)\label{eqn:ml_exp}.
	\end{align}
\end{theorem}

The proof of Theorem~\ref{compressed_exp} can be found in Section~\ref{sec:proof}. For the special choice of $\bA$ as a column-normalized group Hadamard matrix, we can obtain a tight result, which means in the proof of Theorem~\ref{compressed_exp} we show two inequalities, namely that the so-described sensing matrices and maximum likelihood decoding scheme result in an error exponent that satisfies
\begin{align}
\label{mld_ach}
E_{\mathrm{a},\mathrm{ML}}(\beta,\alpha,\mathrm{SNR})\geq\frac{\beta-1+\alpha}{2}\log\left(1+\frac{\alpha\mathrm{SNR}}{2}\right)
\end{align}
and our analysis is {\em ensemble tight} in the sense that 
\begin{align}
\label{mld_cov}
E_{\mathrm{a},\mathrm{ML}}(\beta,\alpha,\mathrm{SNR})\leq\frac{\beta-1+\alpha}{2}\log\left(1+\frac{\alpha\mathrm{SNR}}{2}\right).
\end{align}
Note that we are not claiming an impossibility result (or converse) result over {\em all}  sequences of  sensing matrices $\bA$ and decoding strategies; this would serve as an upper bound on $E_{\mathrm{a}}^*(\beta,\alpha,\mathrm{SNR})$.

In the rest of section, we compare the achievability result in Theorem~\ref{compressed_exp} to an exponent that one can obtain based on a computationally tractable decoding strategy, the modified RDD.

\begin{definition} [Modified RDD error exponent]
	The error exponent with transmission rate $\beta$ and compression rate $\alpha$ 	when we use the modified RDD and set $\bA$ to be a column-normalized group Hadamard matrix with compression rate $\alpha$   is defined as 
	\begin{equation*}
	E_{\mathrm{a},\mathrm{mRDD}}(\beta,\alpha,\mathrm{SNR})=\liminf_{M\rightarrow \infty}-\frac{1}{M}\log P_{\mathrm{a},\mathrm{mRDD}}(\mathrm{err}),
	\end{equation*}
	where $P_{\mathrm{a},\mathrm{mRDD}}(\mathrm{err})$ is the detection error probability when we set $\bA$ to be a column-normalized group Hadamard matrix and the modified RDD in~\eqref{rdd-like} is used.
\end{definition}

\begin{theorem}
	\label{compressed_rdd}
	For $\beta\geq 1$ and $0<\alpha\le 1$ such that  $1/\alpha$ is a natural number, 
	\begin{align}\label{exponent_rdd}
	E_{\mathrm{a}}^*(\beta,\alpha,\mathrm{SNR})&\geq E_{\mathrm{a},\mathrm{mRDD}}(\beta,\alpha,\mathrm{SNR})\\*
	&=\frac{\beta-1}{2}\log\left(1+\frac{\alpha\mathrm{SNR}}{2}\right)\label{eqn:mrdd_exp}.
	\end{align}
\end{theorem}

The proof of Theorem~\ref{compressed_rdd} can be found in Section~\ref{sec:proof}. Our result in Theorem~\ref{compressed_rdd} implies that
\begin{align}
E_{\mathrm{a},\mathrm{ML}}(\beta,\alpha,\mathrm{SNR})> E_{\mathrm{a},\mathrm{mRDD}}(\beta,\alpha,\mathrm{SNR}),
\end{align}
when $\alpha>0$, which means the performance of modified RDD is strictly worse than the maximum likelihood detector. Similarly to the maximum likelihood-based result in Theorem~\ref{compressed_exp}, we analyze the ensemble performance of the modified RDD. This entails showing the achievable performance bound
\begin{align}
\label{rdd_ach}
E_{\mathrm{a},\mathrm{mRDD}}(\beta,\alpha,\mathrm{SNR})\geq\frac{\beta-1}{2}\log\left(1+\frac{\alpha\mathrm{SNR}}{2}\right)
\end{align}
as well as a bound that says that our analysis is tight for the chosen ensemble of sensing matrices and decoding strategy, i.e.,
\begin{align}
\label{rdd_cov}
E_{\mathrm{a},\mathrm{mRDD}}(\beta,\alpha,\mathrm{SNR})\leq\frac{\beta-1}{2}\log\left(1+\frac{\alpha\mathrm{SNR}}{2}\right).
\end{align}

We now briefly comment on the result of Theorem~\ref{compressed_rdd} to Xie, Eldar and Goldsmith's work~\cite{Xie2013}. In~\cite{Xie2013}, using coherence properties of $\bA$, the authors presented an upper bound on the detection error probability using the RDD in~\eqref{rdd} in terms of  the maximum eigenvalue of $\bG$ and the maximum value of $\{\ba_i^T\bA\bA^T\ba_i\}_{i=1}^M$. They showed that the upper bound is $M^{-c}[\pi(1+c)\log M]^{-1/2}$ for some constant $c>0$; this upper bound decays polynomially in $M$. We see from the positivity of $E_{\mathrm{a},\mathrm{mRDD}}(\beta,\alpha,\mathrm{SNR})$ that what the authors of~\cite{Xie2013} obtained  is ostensibly a loose upper bound on the error probability. However, the settings are slightly different (as detailed in  Section~\ref{sec:related}). Our analysis provides a tighter upper bound on the error probability {\color{black}for the detection of $1$-sparse signals} when we consider the modified decoder in~\eqref{rdd-like} in which the absolute value on $\ba_k^T\bu$ is removed. The tighter bound is due to two reasons.
First, we take the randomness of $\bG$ into account in our analysis   (averaging over $\bG$) instead of using the constraints related only to  the {\em maximum} eigenvalue of $\bG$. Second, we choose a sensing matrix (more precisely, a sequence of sensing matrices) that have low coherence (see Lemma \ref{coherence}) instead of using the maximum value of $\{\ba_i^T\bA\bA^T\ba_i\}_{i=1}^M$, a proxy of the coherence. These two reasons allow us to tighten the bound on the error probability in~\cite{Xie2013}, showing that, in fact it decays exponentially fast with exponent given in \eqref{eqn:mrdd_exp}.

\section{Examples and Discussions}

In this section, we first present some numerical examples to compare the behavior of the three error exponents in \eqref{ss}, \eqref{eqn:ml_exp} and \eqref{eqn:mrdd_exp}. We then study how the two error exponents for the approximate sufficient statistics in~\eqref{exponent_mld} and~\eqref{exponent_rdd} depend on the various parameters, namely $\beta$, $\alpha$, and $\mathrm{SNR}$. 

In Fig.~\ref{Fig.examples}, we plot $E_\mathrm{MF}(\beta,\mathrm{SNR})$, $E_{\mathrm{a},\mathrm{ML}}(\beta,\alpha,\mathrm{SNR})$ and $E_{\mathrm{a},\mathrm{mRDD}}(\beta,\alpha,\mathrm{SNR})$ for various values of $\alpha$ and $\beta$ where we note that $1/\alpha$ must be a natural number (i.e., $1/\alpha\in\bbN$).  Specifically, we consider the scenarios in four different settings. Note that we always have
\begin{align}
E_\mathrm{MF}(\beta,\mathrm{SNR})\geq E_{\mathrm{a},\mathrm{ML}}(\beta,\alpha,\mathrm{SNR})>E_{\mathrm{a},\mathrm{mRDD}}(\beta,\alpha,\mathrm{SNR})\notag.
\end{align}
When $\alpha=1$ (no compression), $E_{\mathrm{a},\mathrm{ML}}(\beta,\alpha,\mathrm{SNR})$ reduces to  $E_\mathrm{MF}(\beta,\mathrm{SNR})$. However, $E_{\mathrm{a},\mathrm{mRDD}}(\beta,\alpha,\mathrm{SNR})$ is {\em not equal} to $E_\mathrm{MF}(\beta,\mathrm{SNR})$ when $\alpha=1$. This means that even though there is no compression, the modified RDD suffers from some degradation in terms of the error exponent. Note that this is not due to looseness in our analysis because~\eqref{rdd_cov} shows that our analysis is tight, at least for the sequence of sensing matrices considered. In Fig.~\ref{Fig.sub.1}, we plot  $E_{\mathrm{a},\mathrm{ML}}(\beta,\alpha,\mathrm{SNR})$ and $E_{\mathrm{a},\mathrm{mRDD}}(\beta,\alpha,\mathrm{SNR})$ as a functions of $\alpha$ for fixed $\beta$ and $\mathrm{SNR}$. However, we note that in our analysis and results, the $\alpha$'s that are  permitted are those satisfying $1/\alpha\in\bbN$. To make the curves of $E_{\mathrm{a},\mathrm{ML}}(\beta,\alpha,\mathrm{SNR})$ and $E_{\mathrm{a},\mathrm{mRDD}}(\beta,\alpha,\mathrm{SNR})$ continuous, we linearly interpolate between the points at which $1/\alpha\in\bbN$.  We note that $E_{\mathrm{a},\mathrm{ML}}(\beta,\alpha,\mathrm{SNR})$  tends to zero linearly  when $\alpha$ tends to zero.  In Fig.~\ref{Fig.sub.2}, $E_{\mathrm{a},\mathrm{ML}}(\beta,\alpha,\mathrm{SNR})$ and $E_{\mathrm{a},\mathrm{mRDD}}(\beta,\alpha,\mathrm{SNR})$ increase linearly with $\beta$ for fixed $\alpha$ and $\mathrm{SNR}$. Fig.~\ref{Fig.sub.3} plots $E_{\mathrm{a},\mathrm{ML}}(\beta,\alpha,\mathrm{SNR})$ and $E_{\mathrm{a},\mathrm{mRDD}}(\beta,\alpha,\mathrm{SNR})$ against $\epsilon$  again by linearly interpolating between points corresponding to $1/\epsilon\in\bbN$. When $\alpha=\epsilon$ and $\beta=1+\epsilon$ tend to $0$ and $1$ respectively, the two error exponents behave quadratically in $\epsilon$. In Fig.~\ref{Fig.sub.4}, for fixed $\alpha$ and $\beta$, the two error exponents increase logarithmically with $\mathrm{SNR}$.

\begin{figure*}[h]
	\centering 
	\subfigure[]{
		\label{Fig.sub.1}
		\includegraphics[width=0.48\textwidth]{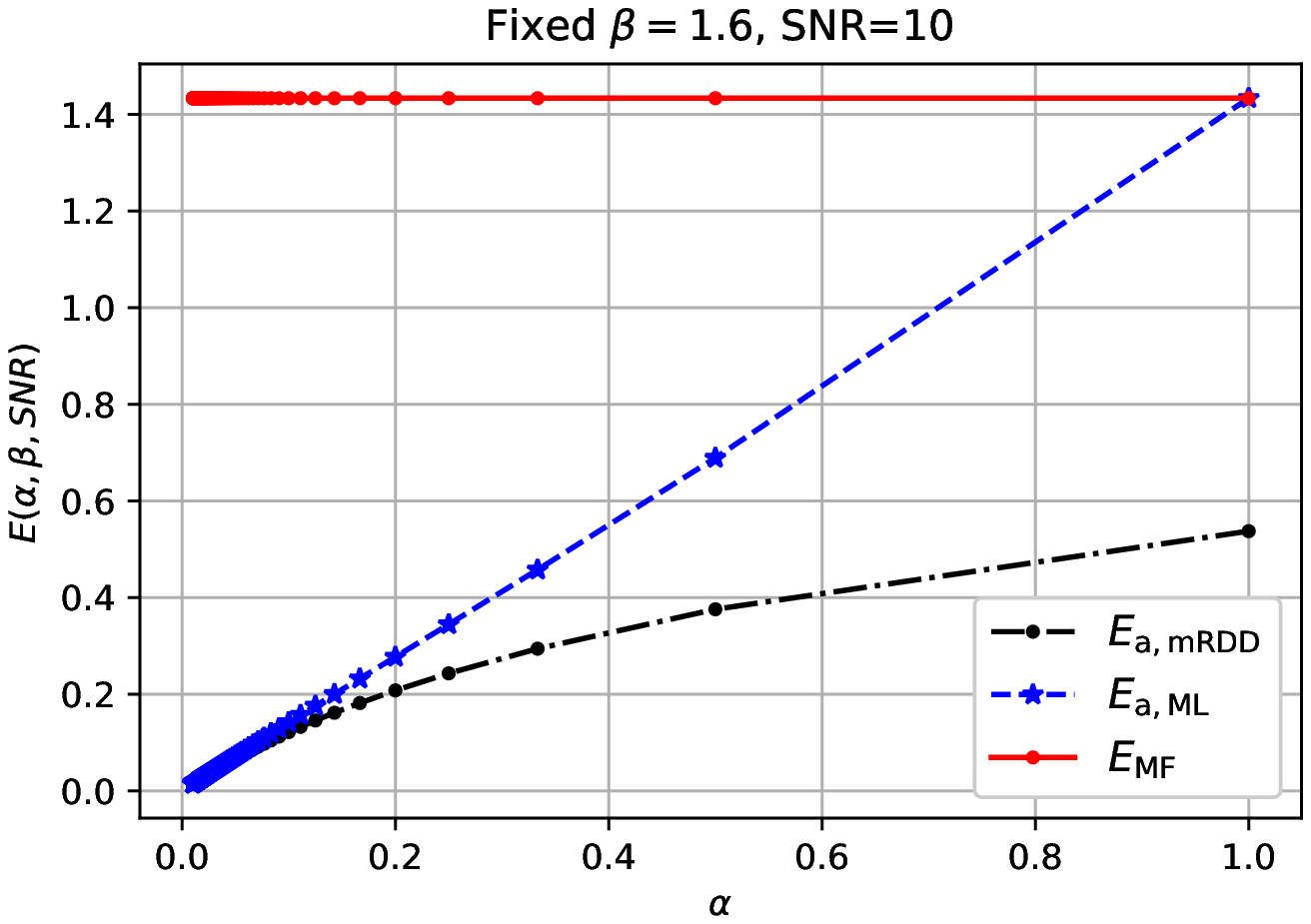}} 
	\subfigure[]{
		\label{Fig.sub.2}
		\includegraphics[width=0.48\textwidth]{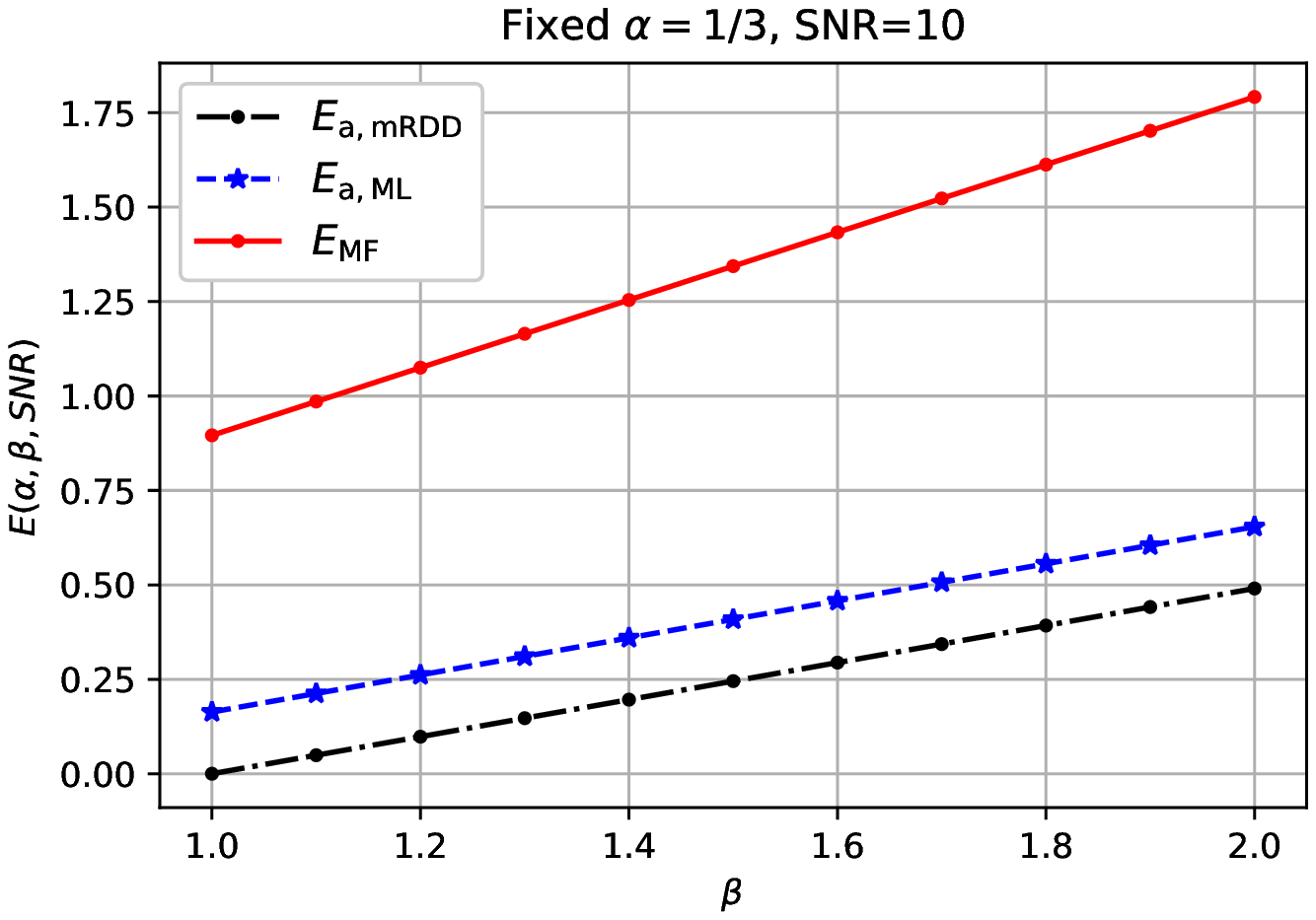}}
	\subfigure[]{
		\label{Fig.sub.3}
		\includegraphics[width=0.48\textwidth]{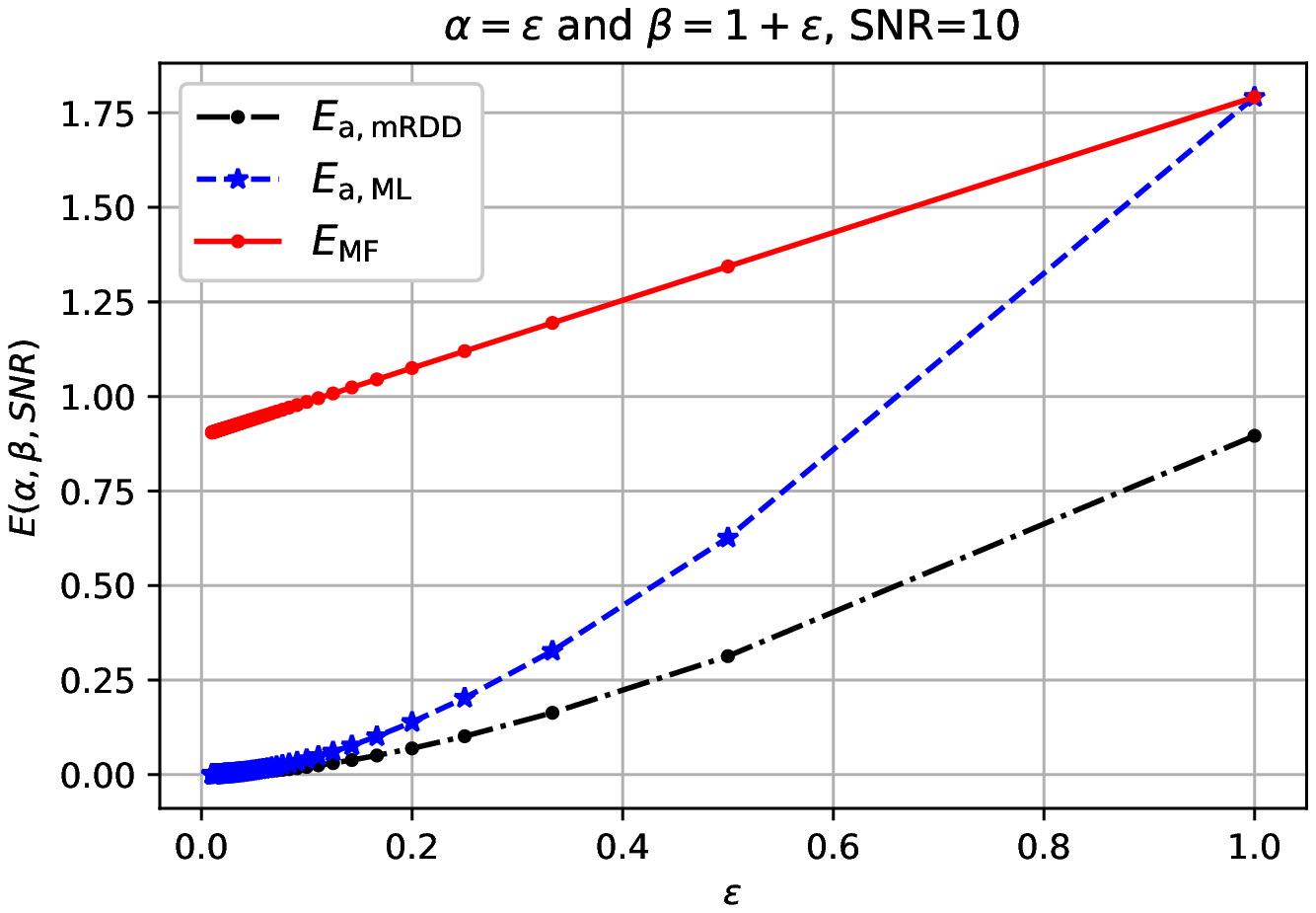}} 
	\subfigure[]{
		\label{Fig.sub.4}
		\includegraphics[width=0.48\textwidth]{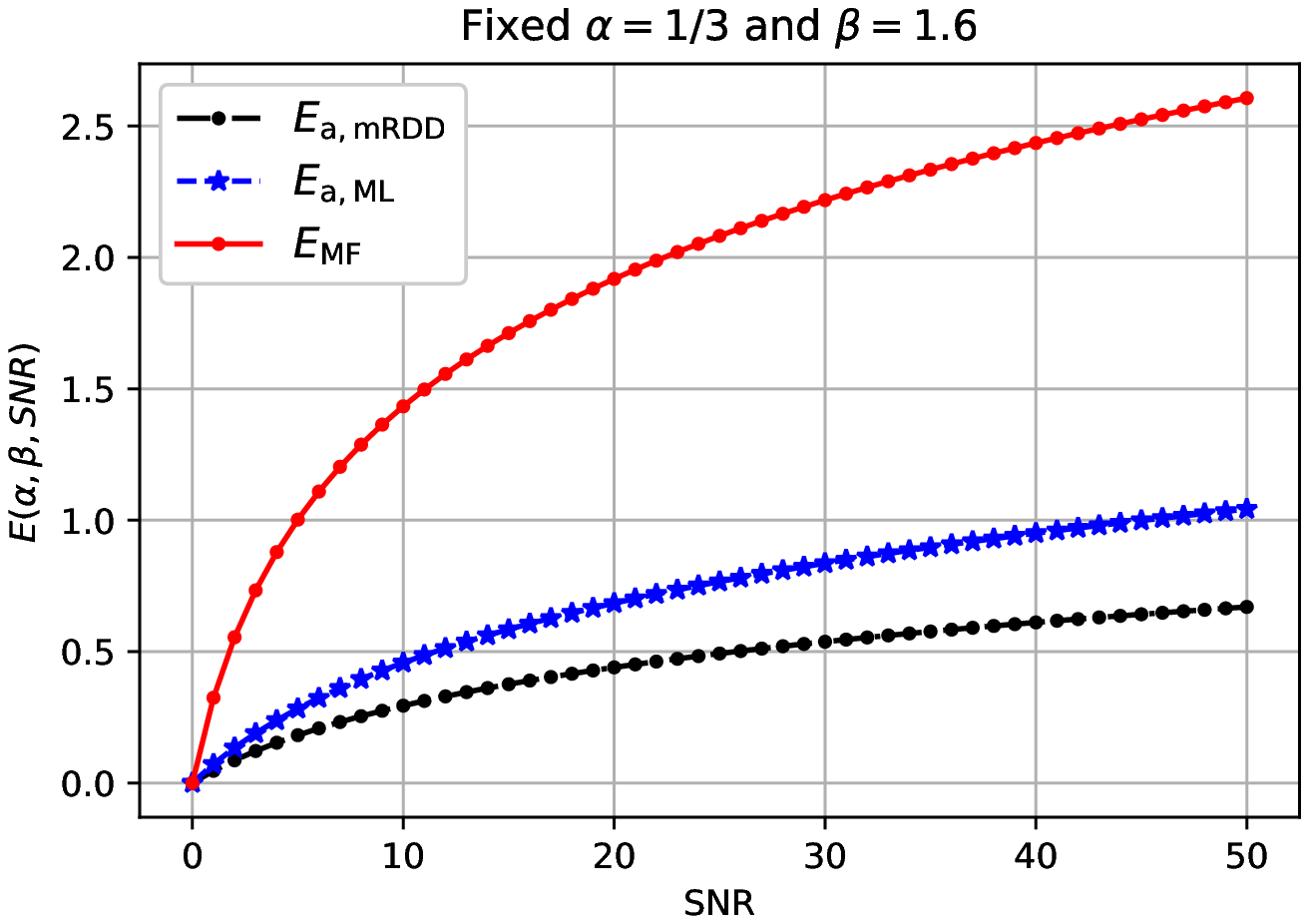}}
	\caption{Performance of the three error exponents. The results are valid only for $\alpha$ such that $1/\alpha\in\mathbb{N}$ but we linearly interpolate to obtain continuous curves.}
	\label{Fig.examples}
\end{figure*}

The following corollary formalizes how the two error exponents we derived depend on the various parameters $\beta$, $\alpha$ and $\mathrm{SNR}$ as they tend to their limiting values. 
\begin{corollary}
	\label{cor}
The following hold:
	\begin{itemize}
	\item[(a)] For a fixed transmission rate $\beta$,  
	\begin{align*}
	E_{\mathrm{a},\mathrm{ML}}(\beta,\alpha,\mathrm{SNR}) &= \calO(\alpha),\\*
	E_{\mathrm{a},\mathrm{mRDD}}(\beta,\alpha,\mathrm{SNR})&=\calO(\alpha), \quad \mbox{as } \alpha\to 0^+, 1/\alpha\in\mathbb{N}.
	\end{align*} 
	\item[(b)] For a fixed compression rate $\alpha$ such that $1/\alpha\in\mathbb{N}$ and SNR, when $\beta=1+\epsilon,\epsilon\in[0,1]$, 
	\begin{align*}
	E_{\mathrm{a},\mathrm{ML}}(\beta,\alpha,\mathrm{SNR}) &=\calO(\epsilon),\\*
	E_{\mathrm{a},\mathrm{mRDD}}(\beta,\alpha,\mathrm{SNR})&=\calO(\epsilon), \quad \mbox{as } \epsilon\to 0^+.
	\end{align*} 
	\item[(c)] For a fixed SNR, when $\alpha=\epsilon, \beta=1+\epsilon, \epsilon\in[0,1]$, and $ 1/\epsilon\in\mathbb{N}$,
	\begin{align*}
	E_{\mathrm{a},\mathrm{ML}}(\beta,\alpha,\mathrm{SNR})&=\calO(\epsilon^2),\\*
	E_{\mathrm{a},\mathrm{mRDD}}(\beta,\alpha,\mathrm{SNR})&=\calO(\epsilon^2), \quad \mbox{as } \epsilon\to 0^+.
	\end{align*} 
	\item[(d)] For a fixed transmission rate $\beta$ and compression rate $\alpha$ such that $1/\alpha\in\mathbb{N}$, 
	\begin{align*}
	E_{\mathrm{a},\mathrm{ML}}(\beta,\alpha,\mathrm{SNR}) &=\calO(\log({\mathrm{SNR}})),\\*
	E_{\mathrm{a},\mathrm{mRDD}}(\beta,\alpha,\mathrm{SNR})&=\calO(\log({\mathrm{SNR}})),\quad\mbox{as } \mathrm{SNR}\to \infty.
	\end{align*} 
	and
	\begin{align*}
	E_{\mathrm{a},\mathrm{ML}}(\beta,\alpha,\mathrm{SNR})&=\frac{\alpha(\beta-1+\alpha)}{2}\cdot \mathrm{SNR}+ o(\mathrm{SNR}),\\*
	E_{\mathrm{a},\mathrm{mRDD}}(\beta,\alpha,\mathrm{SNR})&=\frac{\alpha(\beta-1)}{2} \cdot \mathrm{SNR}+ o(\mathrm{SNR}), \\
	&\hspace{1in}\quad\mbox{as } \mathrm{SNR}\to 0^+.
	\end{align*}
	\end{itemize}
\end{corollary}
Corollary~\ref{cor}(a) implies that for a fixed transmission rate $\beta$ and $\mathrm{SNR}$, when $\alpha$  tends to zero, the two error exponents for the approximate sufficient statistics tend to zero linearly fast. This is not unexpected as increasing the number of approximate statistics linearly provides us with a commensurate amount of information. Corollary~\ref{cor}(b) implies that for a fixed $\alpha$ and $\mathrm{SNR}$, the exponents decrease linearly as $\beta$ tends to 1. This is also natural as we have fewer observations as $T/M$ tends towards $1$ from above.  Corollary~\ref{cor}(c) shows that for a fixed $\mathrm{SNR}$, the two error exponents for the approximate sufficient statistics decreases quadratically as $\alpha$ goes to $0$ and $\beta$ goes to $1$ simultaneously. Corollary~\ref{cor}(d) indicates that the dependences of the three error exponents on  $\mathrm{SNR}$ are similar to that of the capacity of Gaussian channels \cite[Chap.~7]{Cov06}, which is also plausible. In fact, Corollary~\ref{cor}(d) clearly shows the improvement of $E_{\mathrm{a},\mathrm{ML}}(\beta,\alpha,\mathrm{SNR})$ over  $E_{\mathrm{a},\mathrm{mRDD}}(\beta,\alpha,\mathrm{SNR})$ at small $\mathrm{SNR}$ and $\alpha>0$.

\section{Proof of Main Results} \label{sec:proof}
In this section, we provide the proofs of Theorem~\ref{compressed_exp} and Theorem~\ref{compressed_rdd}. We first describe the high-level idea of the proofs. By conditioning on the signals $\{ s_i  \}_{i=1}^M$, we derive a bound on the conditional error probability. Then using certain tools from probability theory~\cite{wishart}, we simplify this conditional probability by averaging over the random signals $\{ s_i  \}_{i=1}^M$. We now introduce some preliminary tools and lemmas that are used extensively in  our proofs.

\subsection{Preliminaries} \label{sec:preli}
Let $\{X_1,X_2,\ldots,X_n\}$ be a sequence of i.i.d.\ zero-mean Gaussian random vectors with  covariance matrix $\Sigma\in\mathbb{R}^{p\times p}$. Let $\mathbf{X}=\sum_{i=1}^nX_iX_i^T$. The probability distribution function of $\mathbf{X}$ is called a {\em Wishart distribution} with covariance matrix $\Sigma$, dimension $n$, and degrees of freedom $p$ and will be denoted as $\mathrm{W}_p(\Sigma,n)$. Lemma~\ref{wish} below characterizes the probability distribution function of $(\bA\bG^{-1}\bA^T)^{-1}$ when $\bA$ is a deterministic matrix with full row rank. The proof can be found in~\cite[Prop.~8.9]{wishart}.
\begin{lemma}
	\label{wish}
	Suppose $\bA\in\bbR^{N\times M}$ with $N\leq M$ is deterministic and has full row rank. Let $T\geq M$ and 
$\bG$ be an $M\times M$  Gram matrix with $[\bG]_{i,j}=\langle s_i(\cdot),s_j(\cdot) \rangle$, where $\{s_i(t)\}_{i=1,t=1}^{M,T}$ are i.i.d.\ zero-mean Gaussian with variance $\calE^2$. Then 
\begin{align}
\label{Wish}
(\bA\bG^{-1}\bA^T)^{-1}\sim\mathrm{W}_N(\calE^2(\bA\bA^T)^{-1},T-M+N).
\end{align}
\end{lemma}
\noindent In~Lemma~\ref{wish}, the assumption $T\geq M$ is used to guarantee the (almost sure) invertibility of the random matrix $\bG$.

We assume that the columns of $\bA$ are normalized, i.e., $\|\ba_i\|_2=1$ for $i\in\{1,2,\ldots,M\}$. The {\em coherence} of $\bA$ is defined as 
\begin{align}
\mu:=\max_{1\leq i\neq j \leq M} |\ba_i^T\ba_j|.  \label{eqn:coherence1}
\end{align}
\begin{definition}
	The $s$-th restricted isometry constant $\delta_s=\delta_s(\bA)$ of a matrix $\bA\in\mathbb{R}^{N\times M}$ is the smallest $\delta\geq0$ such that
	\begin{align*}
	(1-\delta)\|\bx\|_2^2\leq \|\bA\bx\|_2^2\leq (1+\delta)\|\bx\|_2^2
	\end{align*}
	holds for all $s$-sparse vectors $\bx\in\mathbb{R}^M$. Equivalently, $\delta_s$ can be expressed as 
	\begin{align*}
	\delta_s=\max_{S\subset[M],\mathrm{card}(S)\leq s} \|\bA_S^T\bA_S-\bI\|_{2\to2},
	\end{align*}
	where $\|\cdot\|_{2\to2}$ is the spectral norm.
\end{definition}
As usual, we denote $\bb_i$ as a vector with its $i$-th element equal to 1 and others equal to zero. Denote the maximum and minimum eigenvalues of $\bA\bA^T$ as $\lambda_N$  and~$\lambda_1$ respectively. 
\begin{lemma}
	If the $2$-nd restricted isometry constant of $\bA$ is $\delta_2$, then the following inequalities hold: 
	\begin{align}
	\frac{2(1-\delta_2)}{\lambda_N}\leq\|(\bA\bA^T)^{-1/2}\bA(\bb_i-\bb_j)\|_2^2\leq \frac{2(1+\delta_2)}{\lambda_1}.\label{constraint2}
	\end{align}
\end{lemma}
\begin{proof}
	 We have
	\begin{align*}
	\|(\bA\bA^T)^{-1/2}\bA(\bb_i-\bb_j)\|_2^2&\leq\|(\bA\bA^T)^{-1/2}\|_2^2\|\bA(\bb_i-\bb_j)\|^2_2\\*
	&\leq \frac{2(1+\delta_2)}{\lambda_1},
	\end{align*}
	which proves the upper bound in~\eqref{constraint2}. 
	For the lower bound, because
	\begin{align*}
	\|\bA(\bb_i-\bb_j)\|_2^2\leq \|(\bA\bA^T)^{1/2}\|_2^2\|(\bA\bA^T)^{-1/2}\bA(\bb_i-\bb_j)\|^2_2,
	\end{align*}
	we have
	\begin{align*}
	\|(\bA\bA^T)^{-1/2}\bA(\bb_i-\bb_j)\|^2_2&\geq \frac{\|\bA(\bb_i-\bb_j)\|_2^2}{\|(\bA\bA^T)^{1/2}\|_2^2}\\*
	&\geq \frac{2(1-\delta_2)}{\lambda_N},
	\end{align*}
	as desired.
\end{proof}

	Note that based on~\cite[Proposition 6.2]{Foucart}, $\delta_2=\mu$. Thus, \eqref{constraint2} is equivalent to
	\begin{align}
	\frac{2(1-\mu)}{\lambda_N}\leq\|(\bA\bA^T)^{-1/2}\bA(\bb_i-\bb_j)\|_2^2\leq \frac{2(1+\mu)}{\lambda_1}.\label{constraint3}
	\end{align}

Now we choose a column-normalized group Hadamard matrix as our sensing matrix $\bA$. This class of matrices was proposed and analyzed by Thill and Hassibi~\cite{thill}. We introduce how to construct a group Hadamard matrix and state its properties.  For any prime $p$, let $\mathbb{F}_p$ be the ring of integer residual modulo $p$. Let $\mathbb{F}_{p^r}$ be the extension field of $\mathbb{F}_p$ with $p^r$ elements. We denote the elements of $\mathbb{F}_{p^r}$ as $\{x_1,\ldots,x_M\}$ where $M=p^r$. Let $H=\mathbb{F}_{p^r}^{\times}=\mathbb{F}_{p^r}\setminus\{0\}$. It is well known from the theory of  finite fields~\cite{finitefield} that $H$ is isomorphic to the cyclic group of size $p^r-1$. Let $\{a_1,\ldots,a_N\}$ be any subgroup of $H$.  This is a cyclic group of size $N$, where $N$ is a divisor of $p^r-1$. Since $H$ is cyclic, there is a unique subgroup for each $N$, and it consists of the $\big(\frac{p^r-1}{N}\big)^{\mathrm{th}}$ powers in $H$. Thus, if $x$ is a cyclic generator of $H$, we may set $y=x^{\frac{p^r-1}{N}}$ and $a_i=y^i$ for each $i=1,\ldots,N$.
 We let $\bM_p$ be the frame matrix defined as
\begin{align}
\label{frame}
\bM_p:=
\left[
\begin{matrix}
\omega^{\mathrm{Tr}(a_1x_1)}&\omega^{\mathrm{Tr}(a_1x_2)}&\ldots \omega^{\mathrm{Tr}(a_1x_M)}\\
\vdots&\vdots&\vdots\\
\omega^{\mathrm{Tr}(a_Nx_1)}&\omega^{\mathrm{Tr}(a_Nx_2)}&\ldots \omega^{\mathrm{Tr}(a_Nx_M)}\\
\end{matrix}
\right],
\end{align}
where $\omega=e^{\frac{2\pi i}{p}}$ and $\mathrm{Tr}(x)=x+x^p+\ldots+x^{p^r-1}$. If $p=2$, the frame matrix $\bM_p=\bM_2$ reduces to a quantity known as a {\em group Hadamard matrix}. This matrix contains entries that are all equal to $\pm 1$ and hence $\bM_2$ is a real-valued matrix. Then we choose the sensing matrix $$\bA=\frac{1}{\sqrt{N}}\bM_2,$$ i.e., a {\em column-normalized group Hadamard matrix}. Because the set of rows of $\bA$ forms a subset of the rows of an $M\times M$ Hadamard matrix and its columns are normalized, we have the row orthonormality property, i.e,. $$\bA\bA^T=\frac{M}{N}\bI.$$ The following lemma quantifies  the coherence of a column-normalized group Hadamard matrix~\cite[Theorem~8]{thill}. 

\begin{lemma}
\label{coherence}
If $M$ is a power of two, $N$ a divisor of $M-1$, and $\{a_i\}_{i=1}^N$ the elements of the unique subgroup of $\mathbb{F}_{2^r}^{\times}$ of size $N$, then setting $\omega:=e^{\pi i}$, and $\kappa:=\frac{M-1}{N}$ in~\eqref{frame}, the coherence $\mu$ of the column-normalized group Hadamard matrix $\bA$ (as constructed above) satisfies
\begin{align}
\label{eqn:coherence}
\mu\leq \frac{1}{\kappa}\left((\kappa-1)\sqrt{\frac{1}{N}\left(\kappa+\frac{1}{N}\right)}+\frac{1}{N}\right).
\end{align}
\end{lemma}
This lemma states that $\mu$ for $\bA$ can be appropriately upper bounded; in the application of this lemma (in \eqref{eqn:limit} to follow), the upper bound vanishes as $N,M\to\infty$.

\subsection{Proof of Theorem~\ref{compressed_exp}}

To prove Theorem~\ref{compressed_exp}, we must show that~\eqref{mld_ach} and~\eqref{mld_cov} hold. Thus the proof is partitioned into two parts, the achievability in  \eqref{mld_ach} and the ensemble converse in  \eqref{mld_cov}.

\begin{proof}[Proof of~\eqref{mld_ach}]\label{prove:mld_ach}
{\color{black}Given the signals $\{s_i\}_{i=1}^M$, $\bu$ is a Gaussian random vector with mean $\bA\bb$ and covariance matrix $\sigma^2(\bA\bG^{-1}\bA^T)$.} Multiplying $[\bA\bG^{-1}\bA^T]^{-1/2}$ on both sides of~(\ref{eqn:ass}), we obtain that  
\begin{equation}
\tilde{\bu} =[\bA\bG^{-1}\bA^T]^{-1/2}\bA\bb+\bw_0, \notag
\end{equation}
where $\bw_0$ is   white Gaussian noise with zero mean and variance $\sigma^2$. Since under maximum likelihood decoding, the probabilities of error based on $\bu$ and $\tilde{\bu}$ are identical, we bound the detection error probability based on $\tilde{\bu}$ using   Gallager's technique to derive the random coding error exponent~\cite[Chap.~5]{gallager}.


Substituting the probability density function of $\tilde{\bu}$ into  the Gallager bound (the detailed derivation of which is shown in   Appendix~\ref{app:gallger}), we obtain that for $j\neq 1$, and any $0\le \rho\le 1$, 
\begin{align}
\label{gallager2}
&P_{\rm{a},\mathrm{ML}}(\mathrm{err}\mid \bG) \notag\\*
&\leq(M-1)^\rho \exp \bigg\{-\frac{\rho}{2\sigma^2(1+\rho)^2}\notag\\*
&\hspace{2.5cm}\times\left[\| (\bA\bG^{-1}\bA^T)^{-1/2} (\ba_j-\ba_1)\|^2\right.\notag\\*
&\hspace{2.5cm}\left.-(1-\rho)\| (\bA\bG^{-1}\bA^T)^{-1/2} \ba_j \|^2\right] \bigg\}.
\end{align}
We choose $\rho=1$, which turns out to be optimal asymptotically because $M$ grows linearly with $T$.
Then when $\rho=1$, the exponent of~\eqref{gallager2} is
$
\| (\bA\bG^{-1}\bA^T)^{-1/2} (\ba_j-\ba_1)\|^2=
(\ba_j-\ba_1)^T (\bA\bG^{-1}\bA^T)^{-1} (\ba_j-\ba_1).
$
From Lemma \ref{wish}, we have~\eqref{Wish}.
Then using the linear transformation property of the Wishart distribution \cite{wishart}, we have 
$$(\ba_j-\ba_1)^T(\bA\bG^{-1} \bA^T)^{-1}(\ba_j-\ba_1) \sim \sigma_1^2 \chi_{T-M+N}^2,$$ where $\sigma_1^2=\calE^2\|(\bA\bA^T)^{-1/2}\bA(\bb_i-\bb_j)\|_2^2$ and $\chi_{T-M+N}^2$ is the chi-squared distribution with $(T-M+N)$ degrees of freedom.\footnote{For $c>0$, we write $X\sim c\chi_m^2$ to mean that $X/c$ is a chi-squared random variable with $m$ degrees of freedom.}

Let $p(\cdot)$ denote the probability density function of the chi-squared random variable $X\sim\chi_{T-M+N}^2$. We then  obtain
\begin{align}
&P_{\rm{a},\mathrm{ML}}(\mathrm{err})\notag\\*
&\leq \!  M \int_{0}^\infty \! \exp\left\{-\frac{ \mathcal{E}^2\|(\bA\bA^T)^{-\frac{1}{2}}\bA(\bb_i-\bb_j)\|_2^2x}{8\sigma^2}\right\}p(x)\,  \mathrm{d}x\label{im11}\\*
&=\! M \left(1\! +\! \frac{\mathcal{E}^2\|(\bA\bA^T)^{-\frac{1}{2}}\bA(\bb_i-\bb_j)\|_2^2}{4\sigma^2}\right)^{-\frac{ T-M+N }{2}},\label{im1}
\end{align}
where~\eqref{im1} follows from the fact that the integral in~(\ref{im11}) is the moment generating function of a chi-squared random variable. If we choose the sensing matrix as a column-normalized group Hadamard matrix we introduced in Lemma~\ref{coherence}, because $\bA\bA^T=\frac{1}{\alpha}\bI$, we have $\lambda_1=\lambda_N=1/\alpha$. Thus, based on~\eqref{constraint3} and~\eqref{eqn:coherence}, {\color{black} when $M$ is a power of two,} we have
\begin{align*}
&P_{\rm{a},\mathrm{ML}}(\mathrm{err})\\*
&\quad\leq M \left(1+\frac{\alpha\mathrm{SNR}(1-\mu)}{2}\right)^{-\frac{ T-M+N }{2}}\\*
&\quad\leq M\Bigg(1+\frac{\alpha\mathrm{SNR}}{2}\bigg(1-\frac{\alpha}{1-1/M} \\*
&\qquad \times\bigg(\bigg(\frac{1}{\alpha}-\frac{1}{\alpha M}-1\bigg)\sqrt{\frac{1}{M}}+\frac{1}{\alpha M}\bigg)\bigg)\Bigg)^{-\frac{ T-M+N }{2}},
\end{align*}
where we remind the reader that $\mathrm{SNR}=\mathcal{E}^2/\sigma^2$. We can see that 
\begin{align}
\label{eqn:limit}
\lim_{M\to\infty} \frac{\alpha}{1-1/M}\left(\left(\frac{1}{\alpha}-\frac{1}{\alpha M}-1\right)\sqrt{\frac{1}{M}}+\frac{1}{\alpha M}\right)=0,
\end{align}
where, here and in the following, the limit as $M\to\infty$ is understood as being taken along the subsequence indexed by $M=2^i$ for $i\in\bbN$.
Thus,
\begin{align*}
\liminf_{M\rightarrow\infty} -\frac{1}{M} \log P_{\rm{a},\rm{ML}}(\mathrm{err})\geq \frac{\beta-1+\alpha}{2}\log\left(1+\frac{\alpha\mathrm{SNR}}{2}\right),
\end{align*}
as desired.
%
\end{proof} 

\begin{proof}[Proof of~\eqref{mld_cov}]\label{ml_cov}
	Let $\bH=[\bA\bG^{-1}\bA^T]^{-1/2}\bA$ and also let $\bh_i$ be the $i$-th column of $\bH$. Define the error events $\mathcal{E}_i=\{\|\tilde{\bu}-\mathbf{h}_i\|\leq\|\tilde{\bu}-\mathbf{h}_1\|\}$ for $1\le i\le M$. Then the error probability is 
	\begin{align*}
	P_{\rm{a},\mathrm{ML}}(\mathrm{err})=\mathbb{P}\Big(\bigcup_{i\neq 1}\mathcal{E}_i\Big).
	\end{align*}
	We also have
	\begin{align*}
	\mathbb{P}(\mathcal{E}_i)\leq\mathbb{P}\Big(\bigcup_{i\neq 1}\mathcal{E}_i\Big)\leq M \mathbb{P}(\mathcal{E}_i).
	\end{align*}
	We   obtained the upper bound of the error probability based on the union bound, i.e., the Gallager bound with $\rho=1$. Since $\mathbb{P}(\mathcal{E}_i)$ decays exponentially fast in $M$ (and $T$) and $\frac{1}{M}\log M\to0$, the lower bound and upper bound are exponentially tight, which means for lower bound on the error probability, we can just lower bound any one of the error probabilities, say $\mathbb{P}(\mathcal{E}_1)$ (by the symmetry of the events). 
	
	To lower bound $\mathbb{P}(\mathcal{E}_1)$, we first introduce a basic result concerning a lower bound on the complementary cumulative distribution function of the standard Gaussian $Q(x)=\int_x^{\infty} \frac{1}{\sqrt{2\pi}} \exp(-u^2/2)\,\mathrm{d}u$.
	For any $x\geq 0$ and any $\epsilon > 0$, the following inequality holds \cite{cote}:
	\begin{align}
	\label{lower_Q}
	Q(x)\geq\frac{e^{(\pi\epsilon+2)^{-1}}}{2(1+\epsilon)}\sqrt{\frac{\epsilon}{\pi}(\pi\epsilon+2)}\exp\left(-\frac{(1+\epsilon)x^2}{2}\right).
	\end{align}
	Define $c(\epsilon) =	\frac{e^{(\pi\epsilon+2)^{-1}}}{2(1+\epsilon)}\sqrt{\frac{\epsilon}{\pi}(\pi\epsilon+2)}$. We then obtain
	\begin{align*}
	\mathbb{P}(\mathcal{E}_1\mid \bG)&=2Q\left(\frac{\|\mathbf{h}_1-\mathbf{h}_i\|}{2\sigma}\right)\notag\\*
	&\geq c(\epsilon)\exp\left(-\frac{(1+\epsilon)d_{i1}^2}{8\sigma^2}\right),
	\end{align*}
	where $d_{i1}=\|\mathbf{h}_1-\mathbf{h}_i\|=\| (\bA\bG^{-1}\bA^T)^{-1/2} (\ba_j-\ba_1)\|$.
	According to Lemma \ref{wish}, $d_{i1}^2\sim \mathcal{E}^2\|(\bA\bA^T)^{-1/2}\bA(\bb_1-\bb_i)\|_2^2\chi_{T-M+N}^2$. Let $X$ be distributed as $\chi_{T-M+N}^2$ and have probability density function $p(x)$. Define $d_{ij}'=\|(\bA\bA^T)^{-1/2}\bA(\bb_i-\bb_j)\|_2^2 $.
	Then, we obtain
	\begin{align*}
	\mathbb{P}(\mathcal{E}_1)&\geq c(\epsilon) \int_{0}^{+\infty} \exp\left(-\frac{(\epsilon+1)\mathcal{E}^2(d_{1i}')^2 x}{8\sigma^2}\right) p(x)\, \rmd x\notag\\*
	&=c(\epsilon) \left(1+\frac{(\epsilon+1)\mathcal{E}^2(d_{1i}')^2 }{4\sigma^2}\right)^{-\frac{T-M+N}{2}}.
	\end{align*}

	We also choose the sensing matrix to be a column-normalized group Hadamard matrix we introduce in Lemma~\ref{coherence}. Similar to the analysis in the proof of~\eqref{mld_ach}, when $M$ is a power of two, we have
	\begin{align*}
	\mathbb{P}(\mathcal{E}_1)&\geq c(\epsilon) \left(1+\frac{(\epsilon+1)\alpha\mathrm{SNR}(1+\mu)}{2}\right)^{-\frac{T-M+N}{2}}\\*
	&\geq c(\epsilon) \Bigg(1+\frac{(\epsilon+1)\alpha\mathrm{SNR}}{2}\bigg(1+\frac{\alpha}{1-1/M} \\*
	&\quad \times\bigg(\left(\frac{1}{\alpha}-\frac{1}{\alpha M}-1\right)\sqrt{\frac{1}{M}}+\frac{1}{\alpha M}\bigg)\bigg)\Bigg)^{-\frac{(T-M+N)}{2}}.
	\end{align*}
	Note that~\eqref{eqn:limit} holds true and thus
	\begin{align*}
	\limsup_{M\rightarrow\infty}
	-\frac{1}{M}&\log\mathbb{P}(\mathcal{E}_1)\notag\\*
	&\leq \frac{\beta-1+\alpha} {2}\log\left(1+\frac{(\epsilon+1)\mathrm{SNR}\alpha }{2}\right).
	\end{align*}
	\noindent Since $\epsilon>0$ is arbitrary, we can let $\epsilon\rightarrow 0^{+}$ and obtain
	\begin{align*}
	\limsup_{M\rightarrow\infty}-\frac{1}{M}\log\mathbb{P}(\mathcal{E}_1)\leq \frac{\beta-1+\alpha} {2}\log\left(1+\frac{\alpha\mathrm{SNR}}{2}\right).
	\end{align*}
	Because $P_{\rm{a},\mathrm{ML}}(\mathrm{err})\geq \mathbb{P}(\mathcal{E}_1)$, we obtain
	\begin{align*}
	\limsup_{M\rightarrow\infty}-\frac{1}{M}\log P_{\rm{a},\mathrm{ML}}(\mathrm{err})\leq  \frac{\beta-1+\alpha} {2}\log\left(1+\frac{\alpha\mathrm{SNR}}{2}\right),
	\end{align*}
	which completes the proof.
\end{proof}

\subsection{Proof of Theorem~\ref{compressed_rdd}}
To prove Theorem~\ref{compressed_rdd}, we also need to prove that~\eqref{rdd_ach} and~\eqref{rdd_cov} succeed. 

\begin{proof}[Proof of~\eqref{rdd_ach}]\label{proof: rdd_ach}
	Conditioned on $\bG$, the error probability of the modified RDD is 
	\begin{align*}
	P_{\rm{a},\mathrm{mRDD}}&(\mathrm{err}\mid\bG)\\*
	&=\mathbb{P}\left(\bigcup_{i\neq 1}\left\{\mathbf{a}_i^T\mathbf{u}>\mathbf{a}_1^T\mathbf{u}\right\}\mid\mathbf{G}\right)\notag\\*
	&\leq (M-1) \mathbb{P}\left(\mathbf{a}_i^T\mathbf{u}>\mathbf{a}_1^T\mathbf{u}\mid\mathbf{G}\right)\notag\\*
	&=(M-1)\mathbb{P}\left(\mathbf{a}_i^T\mathbf{a}_1+\mathbf{a}_i^T\mathbf{w}>1+\mathbf{a}_1^T\mathbf{w}\mid\mathbf{G}\right)\notag\\*
	&=(M-1)\mathbb{P}\left((\mathbf{a}_i-\mathbf{a}_1)^T\mathbf{w}>1-\mathbf{a}_i^T\mathbf{a}_1\mid\mathbf{G}\right)\notag\\*
	&=(M-1)Q\left(\frac{\tau}{\hat{\sigma}}\right),
	\end{align*}
	where $\tau=1-\mathbf{a}_i^T\mathbf{a}_1$ and $\hat{\sigma}^2=\sigma^2(\mathbf{a}_i-\mathbf{a}_1)^T\mathbf{A}\mathbf{G}^{-1}\mathbf{A}^T(\mathbf{a}_i-\mathbf{a}_1)$.  We introduce an upper bound of the $Q$ function. For any $x>0$, the following inequality holds~\cite{boundQ}:
	\begin{align}
	\label{upper-Q}
	Q(x)\leq \frac{1}{\sqrt{2\pi}x}\exp\left(-\frac{x^2}{2}\right).
	\end{align}

	Now we take the randomness of $\bG$  into consideration. Based on Lemma \ref{wish}, $[\sigma^2(\mathbf{a}_i-\mathbf{a}_1)^T\mathbf{A}\mathbf{G}^{-1}\mathbf{A}^T(\mathbf{a}_i-\mathbf{a}_1)]^{-1}\sim\mathrm{SNR}\cdot \|(\bA\bA^T)^{1/2}\bA(\bb_i-\bb_j)\|_2^{-2}\chi_{T-M+1}^2$. Again let $p(y)$ denote the probability density function of the chi-squared random variable $Y\sim\chi_{T-M+1}^2$.  We also choose the sensing matrix to be a column-normalized group Hadamard matrix as described in Section~\ref{sec:preli}. Because $\bA\bA^T=\bI/\alpha$, we have $(1-\mathbf{a}_i^T\mathbf{a}_1)^2/\|(\bA\bA^T)^{1/2}\bA(\bb_i-\bb_j)\|_2^2=\alpha(1-\ba_i^T\ba_1)/2$, so $\tau^2/\hat{\sigma}^2= Y\cdot \tau\alpha\mathrm{SNR}/2$. Then based on~\eqref{upper-Q}, when $M$ is a power of two, we have
	\begin{align}
	&P_{\rm{a},\mathrm{mRDD}}(\mathrm{err})\notag\\
	&\quad \leq M\int_{0}^{\infty}\sqrt{\frac{1}{\alpha\pi\tau\mathrm{SNR}y}}\exp\left(-\frac{\tau\alpha\mathrm{SNR}y}{4}\right)p(y) \,\rmd y\notag\\
	&\quad=M\int_{0}^{\infty}\sqrt{\frac{1}{\alpha\pi\tau\mathrm{SNR}y}}\exp\left(-\frac{\tau\alpha\mathrm{SNR}y}{4}\right)\notag\\*
	&\qquad\times\frac{y^{\frac{T-M+1}{2}-1}e^{-\frac{y}{2}}}{2^{\frac{T-M+1}{2}}\Gamma\left((T-M+1)/2\right)}p(y)\,\rmd y\notag\\
	&\quad=M\int_{0}^{\infty}\sqrt{\frac{1}{2\alpha\pi\tau\mathrm{SNR}}}\exp\left(-\frac{\tau\alpha\mathrm{SNR}y}{4}\right)\notag\\*
	&\qquad\times\frac{y^{\frac{T-M}{2}-1}e^{-\frac{y}{2}}}{2^{\frac{T-M}{2}}\Gamma\left((T-M)/2\right)}\frac{\Gamma((T-M)/2)}{\Gamma((T-M+1)/2)}\,\rmd y\label{imt1}\\
	&\quad=M\sqrt{\frac{1}{2\alpha\pi\tau\mathrm{SNR}}}\frac{\Gamma((T-M)/2)}{\Gamma((T-M+1)/2)}\notag\\
	&\qquad\times\left(1+\frac{\tau\alpha\mathrm{SNR}}{2}\right)^{-\frac{T-M}{2}}\label{imt11},
	\end{align}
	where in~\eqref{imt1}, we change the random variable $Y\sim\chi_{T-M+1}^2$ to the random variable with distribution $\chi_{T-M}^2$ and~\eqref{imt11} is because the integral   is the moment generating function of the chi-squared random variable.
	Finally, we have
	\begin{align*}
	&P_{\rm{a},\mathrm{mRDD}}(\mathrm{err})\\*
	&\quad\leq M\sqrt{\frac{1}{2\alpha\pi(1-\mu)\mathrm{SNR}}}\frac{\Gamma((T-M)/2)}{\Gamma((T-M+1)/2)}\notag\\
	&\qquad\times\left(1+\frac{(1-\mu)\alpha\mathrm{SNR}}{2}\right)^{-\frac{T-M}{2}}\\
	&\quad=\exp\left(-M\left(\frac{\beta-1}{2}\log\left(1+\frac{(1-\mu)\alpha\mathrm{SNR}}{2}\right)\right)\right.\\
	&\qquad\left.-\frac{\log f(M)}{M}\right),
	\end{align*}
	where $$f(M)=M\sqrt{\frac{1}{2\pi\alpha(1-\mu)\mathrm{SNR}}}\cdot \frac{\Gamma((T-M)/2)}{\Gamma((T-M+1)/2)}.$$ By using Gautschi's inequality~\cite{gautschi},
	\begin{align*}
		\sqrt{\frac{2}{M(\beta-1)}}\leq\frac{\Gamma((T-M)/2)}{\Gamma((T-M+1)/2)}\leq \sqrt{\frac{2}{M(\beta-1-1/M)}},
	\end{align*}
	and because $\mu\ge 0$ and $\mu$ is bounded away from $1$ for all $M$ sufficiently large when $\bA$ is a column-normalized group Hadamard matrix ($\mu$ is a sequence that vanishes as $M$ grows), we have 
	\begin{align*}
	\lim_{M\to\infty}\frac{1}{M}\log f(M)= 0. 
	\end{align*}
	Thus, using~\eqref{eqn:limit} and again the (vanishing) upper bound on $\mu$ in~\eqref{eqn:coherence}, 
	we obtain
	\begin{align*}
	\liminf_{M\rightarrow\infty} -\frac{1}{M} \log P_{\rm{a},\mathrm{mRDD}}(\mathrm{err})\geq \frac{\beta-1}{2}\log\left(1+\frac{\mathrm{SNR}\alpha}{2}\right),
	\end{align*}
	which completes the proof.
\end{proof}

\begin{proof}[Proof of~\eqref{rdd_cov}]
	Conditioned on $\bG$, the error probability of the modified RDD is 
	\begin{align*}
	P_{\rm{a},\mathrm{mRDD}}(\mathrm{err}\mid\bG)&=\mathbb{P}\left(\bigcup_{i\neq 1}\left\{\mathbf{a}_i^T\mathbf{u}>\mathbf{a}_1^T\mathbf{u}\right\}\,\bigg|\, \mathbf{G}\right)\notag\\
	&\geq \mathbb{P}\left(\mathbf{a}_i^T\mathbf{u}>\mathbf{a}_1^T\mathbf{u}\mid\mathbf{G}\right)\notag\\
	&=\mathbb{P}\left(\mathbf{a}_i^T\mathbf{a}_1+\mathbf{a}_i^T\mathbf{w}>1+\mathbf{a}_1^T\mathbf{w}\mid\mathbf{G}\right)\notag\\
	&=\mathbb{P}\left((\mathbf{a}_i-\mathbf{a}_1)^T\mathbf{w}>1-\mathbf{a}_i^T\mathbf{a}_1\mid\mathbf{G}\right)\notag\\
	&=Q\left(\frac{\tau}{\hat{\sigma}}\right),
	\end{align*}
	where $\tau=1-\mathbf{a}_i^T\mathbf{a}_1$ and $\hat{\sigma}^2=\sigma^2(\mathbf{a}_i-\mathbf{a}_1)^T\mathbf{A}\mathbf{G}^{-1}\mathbf{A}^T(\mathbf{a}_i-\mathbf{a}_1)$. Similar to the statements in previous subsection (Proof of~\eqref{rdd_ach}) and based on the lower bound of Q-function in~\eqref{lower_Q}, we have
	\begin{align*}
	P_{\rm{a},\mathrm{mRDD}}(\mathrm{err})&\geq \int_{0}^{\infty} c(\epsilon)\exp\left(-\frac{y(1+\epsilon)\alpha\mathrm{SNR}\tau}{2}\right) p(y)\, \rmd y\notag\\*
	&=c(\epsilon)\left(1+\frac{(1+\epsilon)\alpha\mathrm{SNR}\tau}{2}\right)^{-\frac{T-M}{2}}.
	\end{align*}
	When choosing the sensing matrix to be a column-normalized group Hadamard matrix and $M$ is a power of two, we have
	\begin{align}
	&P_{\rm{a},\mathrm{mRDD}}(\mathrm{err})\notag\\*
	&\quad\geq c(\epsilon)\left(1+\frac{(1+\epsilon)(1+\mu)\alpha\mathrm{SNR}}{2}\right)^{-\frac{T-M}{2}}\label{eqn:int1}\\
	&\quad\geq c(\epsilon)\Bigg(1+\frac{(1+\epsilon)\alpha\mathrm{SNR}}{2}\bigg(1+\frac{\alpha}{1-1/M}\notag\\
	&\qquad \times \bigg(\sqrt{\frac{1}{M}}\left(\frac{1}{\alpha}-\frac{1}{\alpha M}-1\right)+\frac{1}{\alpha M}\bigg)\bigg)\Bigg)^{-\frac{T-M}{2}},\label{eqn:int2}
	\end{align}
	where~\eqref{eqn:int1} is because $\tau\leq 1+\mu$ and~\eqref{eqn:int2} is based on  Lemma~\ref{coherence}.
	Note that~\eqref{eqn:limit} holds true and thus,
	\begin{align*}
	\limsup_{M\rightarrow\infty}
	-\frac{1}{M}\log &P_{\rm{a},\mathrm{mRDD}}(\mathrm{err})\notag\\*
	&\leq \frac{\beta-1} {2}\log\left(1+\frac{(\epsilon+1)\mathrm{SNR}\alpha }{2}\right).
	\end{align*}
	\noindent Since $\epsilon>0$ is arbitrary, we can let $\epsilon\rightarrow 0^{+}$ and obtain
	\begin{align*}
	\limsup_{M\rightarrow\infty}-\frac{1}{M}\log P_{\rm{a},\mathrm{mRDD}}(\mathrm{err})\leq \frac{\beta-1} {2}\log\left(1+\frac{\alpha\mathrm{SNR}}{2}\right),
	\end{align*}
	which completes the proof.
\end{proof}

\section{Conclusion}
\label{sec:conclusion}
In this paper, we constructed approximate sufficient statistics for the $M$-ary hypothesis testing (detection) problem. By using a column-normalized group Hadamard matrix as the sensing matrix and by analyzing the maximum likelihood detector and the modified Reduced Dimensionality Detector of Xie, Eldar, and Goldsmith~\cite{Xie2013}, we obtained two achievable error exponents. We showed that these exponents are ensemble-tight, in the sense that our analysis is tight on the exponential scale. A very pleasing observation that is gleaned from our analysis is that the derived error exponents increase linearly in the compression rate $\alpha$ when $\alpha$ is small, clearly delineating the tradeoff between compression rate and error probability performance. {\color{black}Another appealing conclusion that can be made is that the ML detector performs far better than the modified RDD~\cite{Xie2013} on the error exponent.} However, the former is arguably more difficult to implement and more computationally demanding in practice.

This work, while being a natural offshoot of the vast body of literature in compressed sensing and traditional detection theory, opens several avenues for further investigations. First, a general converse for the error exponent is lacking; the difficulty of this stems  from the complicated statistics of the noise upon processing by an {\em arbitrary} sensing matrix $\bA$. Second, in this work, we restricted ourselves to the regime in which the number of hypothesis $M$ is not larger than the number of observations $T$ (i.e., $M\leq T$). However, in vanilla channel coding, the number of messages scales exponentially in the blocklength. This is a regime of potential interest but we were not able to overcome some technical difficulties for this regime in this work, chiefly because we needed a result similar to Lemma~\ref{wish} for the case in which $M>T$, but in this case {\color{black}$\bG$ is almost surely  singular.} We leave this to future work.

\appendices

\section{Proof of Theorem~\ref{prop:base}} \label{app:thm1}
\subsection{Optimal Error Exponent}\label{base:mle}
 As we assume the prior probabilities for each signal are the same, the optimal error probability is obtained by the maximum likelihood detector. To prove the asymptotic optimality of the matched filter, we first derive the error exponent for the maximum likelihood detector. Denote $\bs_i=[s_i(1),s_i(2),\ldots,s_i(T)]^T$ and $\bS=[\bs_1,\bs_2,\ldots,\bs_M]\in\mathbb{R}^{T\times M}$, we can write~\eqref{eqn:receive} in matrix form as
\begin{align*}
\by=\bS\bb^T+\bz,
\end{align*}
where the observation vector $\by=[y(1),y(2),\ldots,y(T)]^T$ and the noise vector $\bz=[z(1),z(2),\ldots,z(T)]^T$.

Now we detect the true transmitted signal by using the maximum likelihood detector and the error probability is 
\begin{align*}
P^*(\mathrm{err})=\mathbb{P}\left(\argmax_{k\in\{1,2,\ldots,M \}}\bbP(\by\mid k)\neq 1\right).
\end{align*}
For the upper bound on error probability, following similar steps as in Appendix~\ref{app:gallger} with $\bH=\bS$, we have that for any $0\le \rho\le 1$,
\begin{align}
\label{eqn:base_error}
P^*&(\mathrm{err}\mid \bS)\notag\\*
&\leq (M-1)^{\rho} \exp \left\{-\left[\frac{d^2(\bs_i,\bs_1)-(1-\rho)\|\bs_i\|^2}{2\sigma^2(1+\rho)^2} \right]\right\}^{\rho},
\end{align}
where $d^2(\bs_i,\bs_1)=\|\bs_i-\bs_1\|^2$. We choose $\rho=1$ in the bound above. Roughly speaking,  this is asymptotically optimal because $M$ scales linearly with $T$.
When $\rho=1$, the exponent in the upper bound in~\eqref{eqn:base_error} simplifies to $d^2(\bs_i,\bs_1)$ and $d^2(\bs_i,\bs_1)\sim 2\mathcal{E}^2 \chi_{T}^2$. Let $p(\cdot)$ denote the probability density function of the chi-squared random variable $X\sim\chi_{T}^2$. We then obtain
\begin{align*}
P^*(\mathrm{err})&\leq \int_{0}^{\infty}(M-1) \exp \left\{-\left[\frac{2\mathcal{E}^2x}{8\sigma^2} \right]\right\}p(x)\,\rmd x\\*
&=(M-1)\left(1+\frac{\mathrm{SNR}}{2}\right)^{-\frac{T}{2}}.
\end{align*}
Since $\beta= T/M$, we have
\begin{align}
\label{eqn:mleach}
\liminf_{M\to\infty} -\frac{1}{M} \log P^*(\mathrm{err})\geq \frac{\beta}{2}\log\left(1+\frac{\mathrm{SNR}}{2}\right).
\end{align}

For the converse part, we also lower bound the error probability using one error event. Following the same steps in Subsection~\ref{ml_cov} (cf.\ the proof of~\eqref{mld_cov}), we obtain
\begin{align*}
P^*(\mathrm{err}\mid \bS)&=2Q\left(\frac{\|\mathbf{s}_1-\mathbf{s}_i\|}{2\sigma}\right)\notag\\*
&\geq c(\epsilon)\exp\left(-\frac{(1+\epsilon)d^2(\bs_i,\bs_1)}{8\sigma^2}\right).
\end{align*}
Then
\begin{align*}
P^*(\mathrm{err})&\geq \int_{0}^{\infty}c(\epsilon)\exp\left(-\frac{(1+\epsilon)x}{8\sigma^2}\right)p(x)\rmd x\\*
&=c(\epsilon)\left(1+\frac{(1+\epsilon)\mathrm{SNR}}{2}\right)^{-\frac{T}{2}}.
\end{align*}
Since $\beta= T/M$, we have
\begin{align*}
\limsup_{M\rightarrow\infty}-\frac{1}{M}\log P^*(\mathrm{err})\leq \frac{\beta}{2}\log\left(1+\frac{(1+\epsilon)\mathrm{SNR}}{2}\right).
\end{align*}
Since $\epsilon>0$ is arbitrary, we can let $\epsilon\rightarrow 0^{+}$ and obtain
\begin{align}
\label{eqn:mlecon}
\limsup_{M\rightarrow\infty}-\frac{1}{M}\log P^*(\mathrm{err})\leq \frac{\beta}{2}\log\left(1+\frac{\mathrm{SNR}}{2}\right).
\end{align}
 Thus, combining~\eqref{eqn:mleach} and~\eqref{eqn:mlecon} yields $$E^*(\beta,\mathrm{SNR})=\frac{\beta}{2}\log\left(1+\frac{\mathrm{SNR}}{2}\right).$$

\subsection{Matched Filter Error Exponent}
\label{app:asymptotic}
To prove the asymptotic optimality of matched filter, we mainly prove that the error exponent for matched filter is same as the optimal error exponent we derive in Subsection~\ref{base:mle}. We have
\begin{align*}
P_{\mathrm{MF}}(\text{err})&=\bbP\left(\argmax_{j\in\{1,\cdots,M\}} v_j \neq 1\right)\\
&=\mathbb{P}\left(\bigcup_{j\neq 1}\{v_j>v_1\}\right).
\end{align*}
Again we first condition on the signals $\{\bs_i\}_{i=1}^M$, which means that we treat $\bG$ as a deterministic matrix. Then
\begin{align*}
P_{\mathrm{MF}}(\mathrm{err}\mid \bG)&\leq (M-1)\mathbb{P}\left(v_j>v_1\mid\bG\right)\\*
&= (M-1)\mathbb{P}\left(\tilde{v}_j>\tilde{v}_1\mid\bG\right),
\end{align*}
where $\tilde{\bv}=\bG^{1/2}\bb+\bw_0$ is obtained by whitening $\bv$ by multiplying $\bG^{-1/2}$ in both sides of~\eqref{eqn:matched filter}. Following similar steps as in Appendix~\ref{base:mle}, we have
\begin{align*}
P_{\mathrm{MF}}(\mathrm{err}\mid \bG)\leq (M-1)\exp\left\{-\frac{\|\bG^{1/2}(\bb_i-\bb_j)\|^2}{8\sigma^2}\right\}.
\end{align*}
Then because $\|\bG^{1/2}(\bb_i-\bb_j)\|^2=(\bb_i-\bb_j)^T\bG(\bb_i-\bb_j)\sim 2\mathcal{E}^2\chi_T^2$, we obtain
\begin{align*}
P_{\mathrm{MF}}(\mathrm{err})\leq (M-1)\left(1+\frac{\mathrm{SNR}}{2}\right)^{-\frac{T}{2}}.
\end{align*}
So
\begin{align}
\label{eqn:ssuper}
\liminf_{M\rightarrow\infty} -\frac{1}{M} \log P_{\mathrm{MF}}(\mathrm{err})\geq \frac{\beta}{2}\log\left(1+\frac{\mathrm{SNR}}{2}\right).
\end{align}

For the lower bound, following the same steps in Subsection~\ref{ml_cov} (cf.\ the proof of~\eqref{mld_cov}), we obtain
\begin{align*}
P_{\mathrm{MF}}(\mathrm{err}\mid \bG)&=Q\left(\frac{\|\bG^{1/2}(\bb_i-\bb_j)\|}{2\sigma}\right)\\
&\geq c(\epsilon)\exp\left(-\frac{(1+\epsilon)\|\bG^{1/2}(\bb_i-\bb_j)\|}{8\sigma^2}\right).
\end{align*}
Then
\begin{align*}
P_{\mathrm{MF}}(\mathrm{err})\geq c(\epsilon)\left(1+\frac{(1+\epsilon)\mathrm{SNR}}{2}\right)^{-\frac{T}{2}}.
\end{align*}
Since $\epsilon>0$ is arbitrary, we can take the limit of the normalized logarithm and let $\epsilon\rightarrow 0^{+}$ to obtain
\begin{align}
\label{eqn:sslower}
\limsup_{M\rightarrow\infty}-\frac{1}{M}\log P_{\mathrm{MF}}(\mathrm{err})\leq \frac{\beta}{2}\log\left(1+\frac{\mathrm{SNR}}{2}\right).
\end{align}
Combined with~\eqref{eqn:ssuper} and~\eqref{eqn:sslower}, we have 
\begin{align*}
E_{\mathrm{MF}}(\beta,\mathrm{SNR})=E^*(\beta,\mathrm{SNR})=\frac{\beta}{2}\log\left(1+\frac{\mathrm{SNR}}{2}\right),
\end{align*}
which proves Theorem~\ref{prop:base}.

\section{Derivation of (\ref{gallager2})}\label{app:gallger}
Without loss of generality, we suppose $s_1(t)$ is the transmitted signal. Recall that $\bH=[\bA\bG^{-1}\bA^T]^{-1/2}\bA$  and the probability density function of $\tilde{\bu}$ is
\begin{align}
\label{den}
p_1(\tilde{\bu}\mid \bG)=\prod_{i=1}^{N}\frac{1}{\sqrt{2\pi \sigma^2}} e^{-(u_{i}-h_{1,i})^2/2\sigma^2},
\end{align}
where $h_{j,i}$ is the $i$-th element in the $j$-th column of the matrix $\bH$, which is denoted as $\bh_j=[\bA\bG^{-1}\bA^T]^{-1/2}\ba_j$.
In the following, for the sake of brevity, we use the shorthand notation $p_1(\tilde{\mathbf{u}} )$ to mean $p_1(\tilde{\bu}\mid \bG)$. Now we derive an  upper bound on $P_{1}(\mathrm{err})$ using  Gallager's technique to derive the random coding error exponent in channel coding~\cite{gallager}.
For any $\lambda\geq 0$, let
\begin{align*}
\tilde{\mathcal{R}}_1^c=\left\{\tilde{\bu}: \sum_{j\neq 1}\left[\frac{p_j(\tilde{\bu})}{ p_1(\tilde{\bu})}\right]^\lambda\geq 1\right\}.
\end{align*}
Thus, for any $0\le\rho\le1$, the error probability can be upper bounded as
\begin{align*}
P_{1}(\mathrm{err})&\leq \mathbb{P}(\tilde{\mathcal{R}}_1^c)\\*
& \leq \int_{\mathbb{R}^N} \left(\sum_{j\neq 1} \left[\frac{p_j(\tilde{\bu})}{ p_1(\tilde{\bu})} \right]^\lambda \right)^\rho p_1(\tilde{\bu}) \rmd\tilde{\bu}\\*
&= \int_{\mathbb{R}^N} [p_1(\tilde{\bu})]^{(1-\lambda\rho)} \left(\sum_{j\neq 1} [p_j(\tilde{\bu})]^\lambda \right)^\rho \rmd \tilde{\bu}.
\end{align*}
If we let $\lambda=\frac{1}{1+\rho}$,  we obtain 
\begin{align}
\label{gallager}
P_{1}(\mathrm{err})\leq \int_{\mathbb{R}^N}[p_1(\tilde{\bu})]^{\frac{1}{1+\rho}}\left( \sum_{j\neq 1} [p_j(\tilde{\bu})]^{\frac{1}{1+\rho}}\right)^{\rho} \rmd\tilde{\bu}
\end{align}
for any $0\le\rho\le 1$.

 Then we substitute the density function (\ref{den}) into (\ref{gallager}). We then obtain  the chain of inequalities leading to \eqref{eqn:chain} on the  top of this  page. 
 \begin{figure*}[t]
\begin{align}
P_{1}(\mathrm{err})
&\leq \int_{-\infty}^{+\infty} \ldots \int_{-\infty}^{+\infty} \left[ \prod_{i=1}^{n} \frac{1}{\sqrt{(2\pi \sigma^2)}} e^{-\frac{(u_i-h_{1,i})^2}{2\sigma^2}}\right]^{\frac{1}{1+\rho}} \left\{\sum_{j\neq 1} \left[\prod_{i=1}^{n} \frac{1}{\sqrt{2\pi \sigma^2}} e^{-\frac{(u_i-h_{j,i})^2}{2\sigma^2}} \right]^{\frac{1}{1+\rho}} \right\}^\rho \rmd\tilde{\bu}\notag\\
&=\int_{-\infty}^{+\infty} \ldots\int_{-\infty}^{+\infty} \left[\prod_{i=1}^n \frac{1}{\sqrt{2\pi\sigma^2}} e^{-\frac{(u_i^2-2u_i h_{1,i}+h_{1,i}^2)}{ 2\sigma^2(1+\rho) }}\right]  \left\{\sum_{j\neq 1}\left[\prod_{i=1}^n \frac{1}{\sqrt{2\pi\sigma^2}}e^{-\frac{(u_i^2-2u_i h_{j,i}+h_{j,i}^2)}{ 2\sigma^2(1+\rho )}} \right] \right\}^\rho \rmd\tilde{\bu}\notag\\
&=\int_{-\infty}^{+\infty} \ldots \int_{-\infty}^{+\infty} \prod_{i=1}^n \frac{1}{\sqrt{2\pi\sigma^2}}\exp\left\{-\frac{u_i^2-2u_i h_{1,i}+h_{1,i}^2}{2\sigma^2(1+\rho)} -\frac{u_i^2\rho}{2\sigma^2(1+\rho)} \right\} \left\{\sum_{j\neq 1}\prod_{i=1}^n\exp\left\{-\frac{(-2u_i h_{j,i}+h_{j,i}^2)}{2\sigma^2(1+\rho)}\right\}  \right\}^\rho \rmd\tilde{\bu}\notag\\
&= \exp\left\{-\frac{\|\bh_1\|^2\rho}{2\sigma^2(1+\rho)^2} \right\}\int_{-\infty}^{+\infty} \ldots \int_{-\infty}^{+\infty} \prod_{i=1}^n \frac{1}{\sqrt{2\pi\sigma^2}}  
 \exp \left\{-\frac{1}{2\sigma^2}\left(u_i-\frac{h_{1,i}}{1+\rho}\right)^2 \right\}  \notag\\
&\qquad\times \left\{\sum_{j\neq 1}\exp \left\{ -\frac{\|\bh_j\|^2}{2\sigma^2(1+\rho)}\right\} \notag
 \prod_{i=1}^n\exp\left\{ \frac{2u_i h_{1,i}}{2\sigma^2(1+\rho)}\right\} \right\}^\rho \rmd \tilde{\bu}\\
&\leq \exp\left\{-\frac{\|\bh_1\|^2\rho}{2\sigma^2(1+\rho)^2} \right\}  \left\{\sum_{j\neq 1}\exp\left\{-\frac{\|\bh_j\|^2}{2\sigma^2(1+\rho)}\right\}\int_{-\infty}^{+\infty} \ldots \int_{-\infty}^{+\infty}  \prod_{i=1}^n\frac{1}{\sqrt{2\sigma^2\pi}}\exp\left\{-\frac{1}{2\sigma^2}\left(u_i-\frac{h_{1,i}}{1+\rho}\right)^2 \right\} \right.  \notag\\
&\qquad\times\exp\left\{\frac{2u_i h_{j,i}}{2\sigma^2(1+\rho)} \right\} \rmd\tilde{\bu}\Bigg\}^\rho\notag \\
&=\exp\left\{-\frac{\|\bh_1\|^2\rho}{2\sigma^2(1+\rho)^2} \right\} \left\{\sum_{j\neq 1} \exp \left\{ -\frac{\|\bh_j\|^2}{2\sigma^2(1+\rho)}\right\} \exp\left(-\frac{1}{2\sigma^2} \left[\frac{\|\bh_1\|^2}{(1+\rho)^2}-\frac{\|\bh_1+\bh_j\|^2}{(1+\rho)^2} \right]\right)  \right\}^\rho\notag\\
&= \left(\sum_{j\neq 1} \exp\left\{-\frac{1}{2\sigma^2}\left[\frac{d^2(\bh_j,\bh_1)}{(1+\rho)^2} -\frac{(1-\rho)\|\bh_j\|^2}{(1+\rho)^2} \right]\right\} \right)^\rho\notag\\
&=(M-1)^{\rho} \exp \left\{-\left[\frac{d^2(\bh_j,\bh_1)-(1-\rho)\|\bh_j\|^2}{2\sigma^2(1+\rho)^2} \right]\right\}^{\rho}.\label{eqn:chain}
\end{align} \hrulefill
 \end{figure*}
Note that 
\begin{align*}
d^2(\bh_j,\bh_1)=\|[\bA\bG^{-1}\bA^T]^{-1/2}(\ba_j-\ba_1)\|^2
\end{align*}
and
\begin{align*}
\|\bh_j\|^2=\|[\bA\bG^{-1}\bA^T]^{-1/2}\ba_j\|^2.
\end{align*}
This completes the derivation of (\ref{gallager2}).

\bibliographystyle{IEEEtran}
\bibliography{ref}

\begin{thebibliography}{10}
\providecommand{\url}[1]{#1}
\csname url@samestyle\endcsname
\providecommand{\newblock}{\relax}
\providecommand{\bibinfo}[2]{#2}
\providecommand{\BIBentrySTDinterwordspacing}{\spaceskip=0pt\relax}
\providecommand{\BIBentryALTinterwordstretchfactor}{4}
\providecommand{\BIBentryALTinterwordspacing}{\spaceskip=\fontdimen2\font plus
\BIBentryALTinterwordstretchfactor\fontdimen3\font minus
  \fontdimen4\font\relax}
\providecommand{\BIBforeignlanguage}[2]{{%
\expandafter\ifx\csname l@#1\endcsname\relax
\typeout{** WARNING: IEEEtran.bst: No hyphenation pattern has been}%
\typeout{** loaded for the language `#1'. Using the pattern for}%
\typeout{** the default language instead.}%
\else
\language=\csname l@#1\endcsname
\fi
#2}}
\providecommand{\BIBdecl}{\relax}
\BIBdecl

\bibitem{vanTrees}
H.~L. Van~Trees, \emph{Detection, Estimation, and Modulation Theory, Part I:
  Detection, Estimation, and Linear Modulation Theory}.\hskip 1em plus 0.5em
  minus 0.4em\relax John Wiley \& Sons, 2004.

\bibitem{Xie2013}
Y.~Xie, Y.~C. Eldar, and A.~Goldsmith, ``Reduced-dimension multiuser
  detection,'' \emph{IEEE Transactions on Information Theory}, vol.~59, no.~6,
  pp. 3858--3874, 2013.

\bibitem{eldar}
Y.~C. Eldar, \emph{Sampling Theory: Beyond Bandlimited Systems}.\hskip 1em plus
  0.5em minus 0.4em\relax Cambridge University Press, 2015.

\bibitem{eldar2012compressed}
Y.~C. Eldar and G.~Kutyniok, \emph{Compressed Sensing: Theory and
  Applications}.\hskip 1em plus 0.5em minus 0.4em\relax Cambridge University
  Press, 2012.

\bibitem{wainwright}
M.~J. Wainwright, ``Information-theoretic limits on sparsity recovery in the
  high-dimensional and noisy setting,'' \emph{IEEE Transactions on Information
  Theory}, vol.~55, no.~12, pp. 5728--5741, 2009.

\bibitem{wainwright2}
W.~Wang, M.~J. Wainwright, and K.~Ramchandran, ``Information-theoretic limits
  on sparse signal recovery: Dense versus sparse measurement matrices,''
  \emph{IEEE Transactions on Information Theory}, vol.~56, no.~6, pp.
  2967--2979, 2010.

\bibitem{Reeves}
G.~Reeves and M.~Gastpar, ``The sampling rate-distortion tradeoff for sparsity
  pattern recovery in compressed sensing,'' \emph{IEEE Transactions on
  Information Theory}, vol.~58, no.~5, pp. 3065--3092, 2012.

\bibitem{tulino}
A.~M. Tulino, G.~Caire, S.~Verd\'u, and S.~Shamai, ``Support recovery with
  sparsely sampled free random matrices,'' \emph{IEEE Transactions on
  Information Theory}, vol.~59, no.~7, pp. 4243--4271, 2013.

\bibitem{Scarlett}
J.~Scarlett and V.~Cevher, ``Limits on support recovery with probabilistic
  models: An information-theoretic framework,'' \emph{IEEE Transactions on
  Information Theory}, vol.~63, no.~1, pp. 593--620, 2016.

\bibitem{hayashi}
M.~Hayashi and V.~Y.~F. Tan, ``Minimum rates of approximate sufficient
  statistics,'' \emph{IEEE Transactions on Information Theory}, vol.~64, no.~2,
  pp. 875--888, 2017.

\bibitem{Eldar2009}
Y.~C. Eldar, ``Compressed sensing of analog signals in shift-invariant
  spaces,'' \emph{IEEE Transactions on Signal Processing}, vol.~57, no.~8, pp.
  2986--2997, 2009.

\bibitem{thill}
M.~Thill and B.~Hassibi, ``Low-coherence frames from group fourier matrices,''
  \emph{IEEE Transactions on Information Theory}, vol.~63, no.~6, pp.
  3386--3404, 2017.

\bibitem{gallager}
R.~G. Gallager, \emph{Information Theory and Reliable Communication}.\hskip 1em
  plus 0.5em minus 0.4em\relax Springer, 1968, vol.~2.

\bibitem{Cov06}
T.~M. Cover and J.~A. Thomas, \emph{Elements of Information Theory},
  2nd~ed.\hskip 1em plus 0.5em minus 0.4em\relax Wiley-Interscience, 2006.

\bibitem{Horn}
R.~A. Horn and C.~R. Johnson, \emph{Matrix Analysis}.\hskip 1em plus 0.5em
  minus 0.4em\relax Cambridge University Press, 2012.

\bibitem{Foucart}
S.~Foucart and H.~Rauhut, ``A mathematical introduction to compressive
  sensing,'' \emph{Bull. Am. Math}, vol.~54, pp. 151--165, 2017.

\bibitem{coherence}
J.~A. Tropp, ``Just relax: Convex programming methods for identifying sparse
  signals in noise,'' \emph{IEEE Transactions on Information Theory}, vol.~52,
  no.~3, pp. 1030--1051, 2006.

\bibitem{wishart}
\BIBentryALTinterwordspacing
M.~L. Eaton, \emph{The Wishart Distribution (Chapter 8)}, ser. Lecture
  Notes--Monograph Series.\hskip 1em plus 0.5em minus 0.4em\relax Beachwood,
  Ohio, USA: Institute of Mathematical Statistics, 2007, vol.~53, pp. 302--333.
  [Online]. Available: \url{https://doi.org/10.1214/lnms/1196285114}
\BIBentrySTDinterwordspacing

\bibitem{finitefield}
R.~Lidl and H.~Niederreiter, \emph{Structure of Finite Fields}, 2nd~ed.\hskip
  1em plus 0.5em minus 0.4em\relax Cambridge University Press, 1994, p.
  44–75.

\bibitem{cote}
F.~D. C{\^o}t{\'e}, I.~N. Psaromiligkos, and W.~J. Gross, ``{A Chernoff-type
  lower bound for the Gaussian Q-function},'' \emph{arXiv preprint
  arXiv:1202.6483}, 2012.

\bibitem{boundQ}
P.~{Borjesson} and C.-E. {Sundberg}, ``Simple approximations of the error
  function {Q(x)} for communications applications,'' \emph{IEEE Transactions on
  Communications}, vol.~27, no.~3, pp. 639--643, 1979.

\bibitem{gautschi}
D.~W. Lozier, ``Nist digital library of mathematical functions,'' \emph{Annals
  of Mathematics and Artificial Intelligence}, vol.~38, no. 1-3, pp. 105--119,
  2003.

\end{thebibliography}

\end{document}